\documentclass[11pt,draftcls,onecolumn,twoside]{IEEEtran}
\usepackage{graphicx,epsfig}
\usepackage{cite}
\usepackage{amsmath}
\usepackage{amsfonts}
\usepackage{enumerate}
\newtheorem{theorem}{Theorem}

\newtheorem{algo}[theorem]{Algorithm}

\newtheorem{lemma}[theorem]{Lemma}

\newtheorem{proposition}[theorem]{Proposition}

\newcounter{proposition}
\begin{document}

%\title{Pair-wise Markov Random Fields\\ Applied to the Design of Low Complexity MIMO Detectors}
\title{Low Complexity MIMO Detection based on Belief Propagation over Pair-wise Graphs}

%Coordinated Beamforming for the MIMO Broadcast Channel with Limited Feedback}

%\author{\large Chan-Byoung Chae$^{\dagger}$, David Mazzarese$^{\ddagger}$, and Robert W. Heath Jr.$^{\dagger}$\\
%[-19pt] ${}$\\
%[-10pt] ${}$\\
%[-6pt] \small ${}^{\dagger}$Wireless Networking and Communications Group (WNCG)\\
%[-6pt] \small Department of Electrical and Computer Engineering\\
%[-6pt] \small The University of Texas at Austin\\
%[-6pt] \small +1-512-789-9940 (phone), +1-512-471-6512 (fax)\\
%[-6pt] \small \{cbchae, rheath\}@ece.utexas.edu\\
%[-10pt] ${}$\\
%[-6pt] \small $^{\ddagger}$Telecommunication R\&D Center, Samsung Electronics\\
%[-6pt] \small +82-31-279-5210 (phone), +82-31-279-5105 (fax)\\
%[-6pt] \small d.mazzarese@samsung.com}

%Two coloumn version
%\author{{David M1azzarese, Chan-Byoung Chae and Robert W. Heath Jr.}
%\thanks{C.B.Chae and R.W.Heath,Jr. are with Department of Electrical and Computer Engineering, The University of
% Texas, Austin (E-mail : \{cbchae, rheath\}@ece.utexas.edu) D.Mazzarese is with Telecommunication R\&D Center,
%Samsung Electronics Co. Ltd., P.O.Box 105, Suwon, 442-742, Korea
% (E-mail : david.mazzarese@samsung.com)}}

%\author{Chan-Byoung Chae, Kaibin Huang, Nihar Jindal, and Robert W. Heath, Jr. \\
%\thanks{C. B. Chae, K.  Huang and R. W. Heath, Jr. are with Department of Electrical and Computer %Engineering, The University of
% Texas, Austin (E-mail: \{cbchae, rheath\}@ece.utexas.edu, huangkb@mail.utexas.edu). N. Jindal is with Depart. of ECE, University of Minnesota (E-mail : nihar@umn.edu).}}

\author{Seokhyun Yoon,~\IEEEmembership{Member,~IEEE} and Chan-Byoung Chae, \IEEEmembership{Senior Member,~IEEE} \\
%\thanks{This work was supported by Samsung Electronics.}
%\thanks{This work was partly presented at \emph{IEEE VTC2011} and partly supported by Basic Science Research Program through the National Research Foundation of Korea (NRF) funded by the Ministry of Education, Science and Technology (2012R1A1A2038807) and Dankook University Project for funding RICT 2011.} \thanks{S. Yoon is with the School of Electrical, Electronics \& Computer Engineering, Dankook University, Korea (E-mail: syoon@dku.ac.kr). C.-B. Chae is with the School of Integrated Technology, Yonsei University, Korea (E-mail: cbchae@yonsei.ac.kr). }}
\thanks{This work was supported by Basic Science Research Program through the National Research 
Foundation of Korea (NRF) funded by the Ministry of Education, Science and Technology (2012R1A1A2038807) 
and partly by Dankook University Project for funding RICT 2011. 
This work was partially presented at VTC 2011 Spring.}% <-this % stops a space
\thanks{S. Yoon is with the Department of Electronic Engineering, Dankook
University, Yongin-si, Kyunggi-do, Korea, 448-160 (e-mail: syoon@dku.edu)}
\thanks{C.-B. Chae is with the School of Integrated Technology, Yonsei University, Korea (e-mail: cbchae@yonsei.ac.kr)}}

%\small $^{1}$ Wireless Networking and Communications Group (WNCG)\\
%Department of Electrical and Computer Engineering\\
%The University of Texas at Austin, Austin, TX, 78712, USA\\
%Email: \{cbchae, rheath\}@ece.utexas.edu, huangkb@mail.utexas.edu\\
%${}^{2}$
%Department of Electrical and Computer Engineering\\
%University of Minnesota, Minneapolis, MN, 55455, USA\\
%Email: nihar@umn.edu
%\thanks{This work was supported by Samsung Electronics.}}

%\markboth{To be submitted to IEEE JSAC, Rev.~4.0,
%Nov~2007}%
%{Chae et al.: Coordinated Beamforming with Limited Feedback in the MIMO Broadcast Channel}

%\markboth{to be submitted to Electronics Letter,  March~2008}
%{Chae et al.: On the optimality of Linear Multiuser MIMO Beamforming with Many Receive Antennas}

\maketitle \setcounter{page}{1} %\thispagestyle{fancyplain}
%
%\newpage
%\vspace*{-0.9cm}
%
%\renewcommand{\baselinestretch}{1.4}
%

\markboth{\emph{revised, IEEE Trans. on Vehicular Technology},~2013}%
{Pair-wise Markov Random Fields}
%{Pair-wise Markov Random Fields Applied to the Design of Low Complexity MIMO Detectors}

\begin{abstract}
This paper considers belief propagation algorithm over pair-wise graphical models to develop low complexity, iterative multiple-input multiple-output (MIMO) detectors. The pair-wise graphical model is a bipartite graph where a pair of variable nodes are related by an observation node represented by the bivariate Gaussian function obtained by marginalizing the posterior joint probability density under the Gaussian input assumption. Specifically, we consider two types of pair-wise models, the fully-connected and ring-type. The pair-wise graphs are sparse, compared to the conventional graphical model in [18], insofar as the number of edges connected to an observation node (edge degree) is only two. Consequently the computations are much easier than those of maximum likelihood (ML) detection, which are similar to the belief propagation (BP) that is run over the fully connected bipartite graph. The link level performance for non-Gaussian input is evaluated via simulations, and the results show the validity of the proposed algorithms. We also customize the algorithm with Gaussian input assumption to obtain the Gaussian BP run over the two pair-wise graphical models and, for the ring-type, we prove its convergence in mean to the linear minimum mean square error (MMSE) estimates. Since the maximum \emph{a posterior} (MAP) estimator for Gaussian input is equivalent to the linear MMSE estimator, it shows the optimality, in mean, of the scheme for Gaussian input.
\end{abstract}


\newpage
\setcounter{page}{1}

\section{Introduction}
Recent work on multi-input and multi-output (MIMO) detections has mainly been focused on so-called sphere decoding \cite{R1, R2, R3, R4, R5, R6}. Sphere decoding is a two-stage detector in which the channel matrix is first converted into an upper triangular form and, utilizing this structure, a tree search is used for joint data detection. Since the full tree search has the same complexity as maximum likelihood (ML) detection, a sort of reduced search algorithm is applied by limiting the search space, e.g., the number of candidate symbols or radius at each tree search stage. One advantage of sphere decoding is that it can, by choosing an appropriate value of radius or list size, provide a tradeoff between performance and complexity. The performance of sphere decoding has been shown to be quite close to that of ML with a reasonable level of complexity \cite{R6}. To produce soft decisions required for channel decoding, however, the search space cannot be set too small.

Another type of MIMO detector, which has received little attention, is the channel truncation approach \cite{R7,R8,R9,R10}. This approach is also a two-stage detector, where the channel is first converted into a bi-diagonal or, more generally, a poly-diagonal form \cite{R9,R10} and, utilizing the effective channel structure, a trellis search, e.g., the Viterbi algorithm or the forward-backward algorithm \cite{R11,R12}, is used for post-joint detection. The method is similar to the concatenated channel-shortening equalizer and maximum likelihood sequence estimator (MLSE) for the inter-symbol interference channel \cite{R13}. %By employing channel shortening, rather than channel inversion, the noise enhancement that severely affects the performance can be eased, while the amount of interference is limited, allowing maximum likelihood sequence estimation (MLSE) to be implemented with less complexity. 

Another class of MIMO detection worthy of attention is graph based detection \cite{R14,R15,R16,R17,R18,R19,R20}. The approaches are based on the belief propagation (BP) algorithm \cite{R21,R22}. This algorithm has also been extensively studied for the decoding of channel codes, such as the turbo codes and low density parity check codes. In these approaches, the MIMO channel is modeled as a fully-connected bipartite graph, which consists of multiple $N$ observation nodes representing the received signal, multiple $M$ variable nodes representing the hidden data, and the edges connecting the observation nodes with the variable nodes. The resulting graph has the maximal edge degree, i.e., every observation node is connected to every variable node. When applying the BP algorithm~\cite{R21} or the sum-product algorithm~\cite{R22} to such graphs, the complexity is as high as the ML or MAP detector. This is mainly due to the metric computation and the marginalization operation required for the message update at the observation nodes. 

To reduce the computational complexity, the Gaussian BP has been considered in~\cite{R16} and~\cite{R17}, where the input data and messages are all assumed to be Gaussian so that the message and posterior probability can be represented by a pair of mean and variance, resulting in a very simple message update rule. As shown in~\cite{R16} and~\cite{R17}, however, the algorithm converges (though not always) only to the linear minimum mean squared error (LMMSE) solution, which is inferior to the ML detector for non-Gaussian input. On the other hand, \cite{R18} and~\cite{R20} studied complexity reduction via model simplification. Particularly, in~\cite{R18}, it was suggested that to reduce the edge degree some edges in the fully connected bipartite graph should be pruned based on the strength of the channel coefficients. By doing so, not only is the number of messages reduced, but also the marginalization operation on observation nodes can be performed at a much lesser cost. Reduction in the marginalization cost is exponential with the edge-degree reduction, resulting in far less complexity than ML. The problem here, however, is that the performance loss can be more severe with the edge-degree reduction. 

Other interesting graph-based approaches are those in~\cite{R20, R31, R32, R33}, based on the pair-wise Markov random field (MRF)~\cite{R23}. In MRF, we have only one type of node representing the hidden data and the edges reflecting the local dependency among them. The local dependency is represented by potential functions and, specifically in pair-wise MRF, they are functions of one or two variables. In fact, as noticed in~\cite{R20, R31} and~\cite{R33} (also in~\cite{R16} and~\cite{R19}), a multivariate Gaussian function can be decomposed into a product of functions of one or two variables resulting in a fully connected pair-wise MRF. On the other hand, in~\cite{R20}, noticing that BP may not work well for a loopy graph, the authors proposed a tree approximation on the basis of Kullback-Leibler distance (KLD) optimality criterion. In \cite{R32}, the same authors have proposed using the potential functions obtained by two dimensional projection and the expectation-maximization (EM) algorithm based post detection. Bit-based probabilistic data association is another approach to low complexity MIMO detection especially for higher order QAM. In~\cite{R34}, a matrix representation is introduced to represents symbol mapping, by which it can be considered as a linear processing and can be combined as part of MIMO channel giving us a room for complexity reduction for higher order QAM.

In this paper, we investigate a similar approach to the pair-wise MRF based MIMO detector, but with different formulation, i.e., instead of using the potential functions obtained from the direct decomposition of multivariate Gaussian function \cite{R20, R32} or from the two dimensional projection in \cite{R32}, we propose using the functions obtained by marginalizing the posterior joint probability density under the Gaussian input assumption. As implicated in \cite{R31,R33}, the corresponding bipartite graph has an edge degree of only two and the proposed scheme has much less complexity than that of ML/MAP. In addition to the fully connected pair-wise graph, we also consider the ring-type pair-wise graph. The proposed scheme can be regarded as an edge pruning technique, similar to the one in~\cite{R18}. Unlike that of~\cite{R18}, however, the pruning is performed by a linear transformation and the performance degradation compared to the ML/MAP detector is shown to be reasonable with an edge degree of two.

%however, performance degradation is shown to be reasonable and, as shown in the simulation results, the bit error performance is shown to have close-to ML performance for non-Gaussian input. A similar work has recently been proposed in~\cite{R31}, where the scheme can be understood by the fully connected MRF with potential functions obtained by two dimensional projection. This is similar to the bi-truncation in zero-forcing (ZF) sense, discussed also in~\cite{R9}.

This paper is organized as follows. In the next section, we briefly review the ML/MAP and the graph-based approach to MIMO detection. In Section~\ref{Sec:III}, the proposed iterative detection algorithm is presented based on the fully-connected and ring-type pair-wise models, respectively, for non-Gaussian input. In Section IV, we customize the proposed algorithms under Gaussian input assumption (Gaussian BP), and discuss its convergence property. The performance is extensively evaluated and compared via link-level simulations in Section V and, finally, in Section VI, the concluding remarks are given.

\section{System model, MAP and Graph-based Detection}
\label{Sec:II}
%In this section, we briefly review the channel model and some of the graph-based detectors related to our work. 
%\subsection{System Model, MAP, and Graph-based Detection}
$\emph{System Model}$: A Gaussian MIMO system with an $N\times M$ channel matrix ${\pmb H} (N \geq M)$ is modeled as
\begin{align}
{\pmb y} = {\pmb H} {\pmb x} + {\pmb n} = \sum_{k=1}^M {\pmb h}_k x_k + {\pmb n} \nonumber
\end{align}
where ${\pmb x}$ is an $M \times 1$ transmitted data symbol vector, ${\pmb n}$ is an $N \times 1$ noise vector, ${\pmb y}$ is an $N \times 1$ received signal vector and ${\pmb h}_m$ is the $m$th column of ${\pmb H}$. A symbol ${\pmb n}$ is assumed to be complex Gaussian with mean ${\pmb 0}$ and covariance ${\mathbb E}[{\pmb n} {\pmb n}^H] = \sigma^2 {\pmb I}$ and the transmitted data symbol vector ${\pmb x}$ is assumed to have mean ${\pmb 0}$ and covariance matrix ${\mathbb E}[{\pmb x} {\pmb x}^H] = {\pmb I}$, where ${\mathbb E}(\cdot)$ denotes expectation. In practice, each element of ${\pmb x}$ is usually a $2^m$-ary symbol drawn from a finite alphabet set $\Xi$ of size $2^m$ such as QPSK and 16-QAM.

$\emph{MAP detection}$: The maximum \emph{a posteriori} (MAP) detector selects ${\pmb x}$ that maximizes the \emph{a posteriori} likelihood
\begin{align}
p ({\pmb x} | {\pmb y}) = \frac{p({\pmb y}|{\pmb x} ) p({\pmb x})}{p({\pmb y})} \label{eq:1}
\end{align}
where 
\begin{align}
&p({\pmb y}|{\pmb x} ) = {\mathcal {CN}} \left({\pmb y}; {\pmb {Hx}}, \sigma^2 {\pmb I} \right) \label{eq:3}\\
&p({\pmb x}) = \prod_{j=1}^M p(x_j). \nonumber
\end{align} 
with ${\mathcal {CN}} ({\pmb y}; {\pmb \mu}, {\pmb C})$ representing a multivariate complex Gaussian probability density function (PDF) of mean ${\pmb \mu}$ and covariance ${\pmb C}$ defined as
\begin{align}
{\mathcal{CN}} ({\pmb y}; {\pmb \mu}, {\pmb C}) \equiv 
\frac{1}{  (\pi)^M \text{det} {\pmb C}  }\exp \left( -\frac{1}{2} ({\pmb y} - {\pmb \mu})^H {\pmb C}^{-1} ({\pmb y} - {\pmb \mu})  \right) \nonumber
\end{align}
where the superscript $H$ denotes Hermitian transpose. 
The search space of the MAP is an $M$-dimensional space, $\Xi^M$, and the complexity is $\mathcal{O}(2^{mM})$. When using concatenated channel coding and MIMO, a MIMO detector is required to produce soft-decision values, i.e., log-likelihood ratio (LLR). Denoting the $j$th data symbol as $x_j(b_{j1}, b_{j2}, \cdots, b_{jm})$, where $b_{j,k}$ is the $k$th bit contained in $x_j$. Then, LLR of $b_{jk}$ can be obtained by first marginalizing $p({\pmb x} |{\pmb y})$ over ${\pmb x}  \backslash x_j = (x_1, x_2, \cdots, x_{j-1}, x_{j+1}, \cdots, x_M)$ to get
\begin{align}
\begin{split}
&p (x_j = x | {\pmb y}) = \sum_{{\pmb x} \backslash x_j \in \Xi ^{M-1}} p (x_1, x_2, \cdots, x_M | {\pmb y}  )\\
&= A \cdot \sum_{{\pmb x} \backslash x_j \in \Xi ^{M-1}} p ({\pmb y} | x_1, x_2, \cdots, x_j = x, \cdots, x_M) \prod_{k\neq j} p(x_k) \label{eq:4}% check please
\end{split}
\end{align}
where $A = p^{-1} ({\pmb y})$ is a normalizing constant, and $x_j$'s are assumed to be independent of each other. In (\ref{eq:4}), $p(x_j)$ is the \emph{a priori probability} of $x_j$, which is assumed to be uniformly distributed, i.e., $p(x_j)=1/2^m$ for a modulation size of $2^m$. The LLR for each bit is then computed as
\begin{align}
{\text{LLR}} (b_{j,k}) = \log \left( \frac{p(b_{j,k} = 0 | {\pmb y})}{p(b_{j,k} = 1 | {\pmb y})}    \right) = \log \left( \frac{\sum_{\text{all}~ x_j: b_{jk}=0} p(x_j | {\pmb y} ) }{\sum_{\text{all}~ x_j: b_{jk}=1} p(x_j | {\pmb y})}  \right). \label{eq:5}
\end{align}
The prohibitive complexity,when $m$ and $M$ are large, comes from the marginalization operation in (\ref{eq:4}).

%Since the noise covariance is given by a diagonal matrix, the marginalization can also be performed for each element of ${\pmb y}$, i.e.,

%\begin{align}
%\begin{split}
%p(x_j = x | y_k) &= \sum_{{\pmb x} \backslash x_j \in \Xi^{M-1}} p(x_1, x_2, \cdots, x_M | y_k) \\
%&\propto p(x_j) \sum_{{\pmb x} \backslash x_j \in \Xi^{M-1}}  p (y_k | x_1, x_2, \cdots, x_j = x, \cdots x_M) \prod_{k\neq j} p(x_k) \label{eq:6}
%\end{split}
%\end{align}
%where, 
%\begin{align}
%p (y_k | x_1, x_2, \cdots, x_j = x, \cdots x_M) = p(y_k | {\pmb x} = {\mathcal{CN}} \left(y_k; \sum_{j=1}^{M} h_{kj} x_j, \sigma^2   \right) \label{eq:7}
%\end{align}
%such that $p({\pmb y} | {\pmb x}) = \prod_k p(y_k| {\pmb x})$. Note that the marginalization in (\ref{eq:6}), as well as in (\ref{eq:4}), is performed over $M-1$ dimensional space. Since the number of lattice points in $M-1$ dimensional space is $2^{m(M-1)}$ and the marginalization must be performed for the total number of $2^m$ states of $x_j$, the complexity remains the same as that of MAP. 

\emph{Graph-based detection (BP over fully-connected bipartite graph)}: The MAP detections in (\ref{eq:4}) is useful for turbo equalization~\cite{R27,R28}, where one can find a vast amount of literature showing the validity of iterative MIMO detection and channel decoding. Although turbo equalization is not our main focus in this paper, it is worthy of paying attention to the iterative detection as shown in~\cite{R18}, i.e., the BP over the fully-connected bipartite graph. In fact, the MAP detection in (\ref{eq:4}) can be regarded as a BP that is run over the singly connected factor graph as shown in Fig.~1(a), where each variable node, representing a data symbol, first passes \emph{a priori} information to the observation node labeled by the received vector, ${\pmb y}$. The observation node then provides each variable node with the corresponding \emph{a posteriori} likelihood by computing the marginalization in (\ref{eq:4}). Since the graph is singly connected and all variable nodes are connected via one observation node, the BP over this graph will surely converge, in one iteration, to the correct \emph{a posteriori} probability. The graph-based detection in~\cite{R18}, on the other hand, is a BP over the fully connected bipartite graph as shown in Fig.~1(b), where the marginalization is performed separately for each observation node and they are then combined to produce the belief and the extrinsic information on each data symbol. 
The algorithm in \cite{R18} can be summarized as follows. \newline

\textsf{BP 1 over the fully-connected factor graph~\cite{R18}}

\noindent \texttt{\small For given {\it a priori} probability of $x_j$, which is typically assumed to be uniformly 
\\distributed, i.e., $p(x_j) = 1/2^m$ for a modulation size of $2^m$}

$\texttt{(1)~Initialization:}$ 
\begin{align}
\lambda_{j \rightarrow i} (x_j) = p (x_j). \nonumber
\end{align}

$\texttt{(2)~Observation~node~computation:}$
%\begin{align}
%\pi_{i \rightarrow j} (x_j = x) = A \cdot \sum_{{\pmb x} \backslash x_j \in \Xi ^{M-1}} p(y_k | x_1, x_2, \cdots, x_j = x, \cdots, x_M) \prod_{k\neq j} \lambda_{k\rightarrow i} (x_k) \label{eq:9}
%\end{align}
\begin{align}
\pi_{i \rightarrow j} (x_j ) = A \cdot \sum_{{\pmb x} \backslash x_j \in \Xi ^{M-1}} p(y_i | x_1, x_2, \cdots, x_M) \prod_{k\neq j} \lambda_{k\rightarrow i} (x_k). \label{eq:9}
\end{align}

$\texttt{(3)~Belief~update:}$
\begin{align}
b(x_j) =  \prod _{k=1}^{M} \pi_{k \rightarrow j} (x_j). \label{eq:10}
\end{align}

$\texttt{(4)~Variable~node~computation:}$
\begin{align}
\lambda_{j \rightarrow i} (x_j) =  \prod_{k\neq j} \pi_{k \rightarrow j} (x_j) = \frac{b(x_j)} { \pi_{i \rightarrow j} (x_j)}. \label{eq:11}
\end{align} 
\noindent \texttt{\small The message update (\ref{eq:9})-(\ref{eq:11}) are repeated by a pre-defined number or until the \\belief does not change any more.} 

\indent Note that, in (\ref{eq:9}), $p({y_i}|{x_1},{x_2},..,{x_j} = x,...,{x_M})$ is given by
\begin{align}
 p(y_i| x_1,x_2,..,x_M) = \mathcal{CN} \left( y_i; \sum\nolimits_{j=1}^{M} h_{ij}x_j,\sigma ^2 \right) \nonumber % eq08
\end{align}
and, by combining (\ref{eq:9}) and (\ref{eq:10}), we see that, at the first iteration
\begin{align}
 b({x_j}) \propto \prod\nolimits_{k = 1}^M p( x_j | y_k ) \nonumber % eq09
\end{align}
which is certainly different from  $p(x_j | \pmb{y})$ in (\ref{eq:4}). That is, in BP algorithm, we first marginalize 
$p(\pmb{x} | y_k)$ for each received signal $y_k$ to obtain $p(x_j | y_k)$  and, then, the belief is obtained 
by their product, while, in MAP, we just marginalize $p( \pmb{x} | \pmb{y})$, once and for all. 
Note also that since the marginalization in (\ref{eq:9}) is performed over $M-1$ dimensional space and must be 
performed for the total number $2^m$ states of $x_j$, the complexity for one iteration is the same as that 
of MAP and the total complexity is multiplied by the number of iteration resulting in far complex computation 
than that of MAP detection. Regardless of its complexity, however, it provides a base structure for the 
development of low complexity detector.

%\subsection{Complexity Reduction via Edge Pruning}
$\emph{Complexity Reduction via Edge Pruning}$: To reduce the computational burden of the marginalization in (\ref{eq:9}) for non-Gaussian input, \cite{R18} proposed  pruning some edges of which the corresponding variable and observation nodes are weakly coupled together, e.g., those variable-observation node pairs with small values of $|h_{jk}|$. By using only (edge degree) $d_{\text f} < M$ edges per observation node (i.e., pruning $M-d_{\text f}$ edges), the complexity is reduced by a factor of $1/2^{m(M-d_{\text f})}$ relative to the ML/MAP or the BP 1 of complexity $\mathcal{O}(2^{mM})$. Here $d_{\text f}$ is the edge degree. The problem with this scheme is that $d_{\text f}$ must be large enough to ensure a reasonable performance, as shown in~\cite{R18}. 

\section{Detection Algorithm based on Pair-wise Graphical Models}
\label{Sec:III}
In this section, we develop low complexity iterative MIMO detection algorithms based on the pair-wise graphical models. We consider two types, namely, the fully-connected and the ring-type, and derive the corresponding BP algorithms that work for non-Gaussian input. As will be shown below, BP over the ring-type pair-wise graph is, with a slight difference, effectively equivalent to the one in~\cite{R10}.

\subsection{BP based on pair-wise Markov Random Field}
Our starting point is the BP algorithm based on pair-wise Markov random field (MRF) in~\cite{R19},~\cite{R20} and~\cite{R33}. MRF  is an undirected graph that describes local dependencies among a set of random variables. In MRF, the joint PDF of all random variables involved can be represented by a product of the joint PDF of each clique.\footnote{A clique in a graph is defined by a set of nodes having full-connection to each other.} The pair-wise MRF means that a joint PDF (of all variables involved) is represented by a product of joint PDFs with only two variables corresponding to an edge connecting any two neighbors. 
%Fig.~\ref{Fig:4} shows two types of pair-wise MRF with four variables. 
Let $V= \{1,2, \cdots,M\}$ be the set of nodes in the MRF corresponding to the random variables $x_1, x_2, \cdots, x_M$, respectively, and let $E$ be the set of all edges connecting these nodes. For a compact expression, we also denote the edge connecting nodes $j$ and $k$ as $e(j,k)$ and the set of neighbors of the $j$th node as $V(j)$. In pair-wise MRFs, the \emph{a posteriori} joint function $p(x_1, x_2, \cdots, x_M |{\pmb y} $) is modeled by a product of pair-wise potential functions~\cite{R17,R23}, e.g.,
\begin{align}
\hat{p} (x_1, x_2, \cdots, x_M |{\pmb y} ) = A \cdot \prod_{i \in V} \psi_i (x_i) \prod_{(i,j): e(i,j) \in E} \phi_{ij} (x_i, x_j), \label{eq:16}
\end{align}
where $\psi(x_i)$ is self-potential assigned to each node and $\phi(x_i, x_j)$ is the edge potential assigned to each edge. Such modeling based on a pair-wise MRF can also facilitate the marginalization to finally obtain the marginal distribution for each random variable.
%, which is exactly what MAP detector does as shown in (\ref{eq:4}) and has a complexity exponentially increasing with the number of random variables. With pair-wise MRF, on the other hand, it can be performed by a ``successive marginalization" if the pair-wise MRF has no loop, or, more generally, by recursive message passing, i.e., BP algorithm. 
Denoting the (incoming) message from the $i$th to the $j$th node as $\pi_{i \rightarrow j} (x_j)$, the BP through the pair-wise MRF can be described as \cite{R17}
\begin{align}
\pi_{i \rightarrow j} (x_j) = \alpha \sum_{x_i \in \Xi} \psi_i (x_i) \phi_{ij} (x_i, x_j) \prod_{k \in V(i) \backslash j} \pi_{k\rightarrow i} (x_i) \label{eq:29}
\end{align}
where, $\alpha$ is the normalizing constant, $V(i)\backslash j$ is the set of neighbors of node $i$ excluding node $j$. Note that we follow the convention in \cite{R17},\cite{R19}, and\cite{R20} to describe the message passing over a MRF, where only one type of node, say the variable nodes, exist and the message flies between these variable nodes. When we use a bipartite graph as shown in Fig.~1(b), we need to define two types of messages, i.e., one from variable node to observation node and the other from observation node to variable node, which can be easily obtained by dividing (\ref{eq:29}) into two separate steps, i.e., $\lambda_{i \rightarrow j} (x_i) = \prod_{k \in N(i)\backslash j} \pi_{k \rightarrow i} (x_i)$ (variable-to-observation node message) and $\pi_{i \rightarrow j} (x_j)  = \alpha \sum_{x_i \in \Xi} \psi_i (x_i) \phi_{ij} (x_i, x_j) \cdot \lambda_{i \rightarrow j} (x_i)$ (observation-to-variable node message).

In (\ref{eq:29}), the incoming messages are combined first to produce the extrinsic information, $\prod_{k\in V(i) \backslash j} \pi_{k\rightarrow i} (x_i)$, and they are then ``translated" by the potential function, $\psi_i (x_i) \phi_{ij} (x_i, x_j)$. The belief on the variable, $x_j$, is given by
\begin{align}
b(x_j) = \prod_{k \in V(i)} \pi_{k \rightarrow j} (x_j). \label{eq:30}
\end{align}

\noindent The potential functions in  (\ref{eq:16}) is given by a fatorization of the joint {\it a posteriori} probability. Specifically, in \cite{R17, R34}, the potential function is obtained by decomposition of multivariate Gaussian function, i.e.,  
\begin{align}
\begin{split}
&\phi_{i,j} (x_i, x_j) = A_{ij} \exp \left(- \frac{1}{\sigma^2} \text{Re} [x^*_i R_{ij} x_j ] \right)  \label{eq:a19} \\
&\psi_i (x_i) = A_i \exp \left(-\frac{1}{\sigma^2} \text{Re} [x_i^* y'_j - R_{ii} |x_i|^2 ]   \right) 
\end{split}
\end{align}
where $R_{ij} = {\pmb h}_i^H {\pmb h}_j$ and $y_j' = {\pmb h}_i^H {\pmb y}$, and $^*$ denotes complex conjugate. In fact, such decomposition gives us a fully connected pair-wise MRF and is exact in the sense that (\ref{eq:16}) with the functions in (\ref{eq:a19}) is exactly the same as the joint Gaussian PDF. It has been shown in~\cite{R16} and \cite{R17} that, with (\ref{eq:a19}), the BP over the fully connected pair-wise MRF results in the MMSE solution if it converges (though the convergence is not always guarateed for arbitrary channel matrices). Most of all, however, it does not work well for non-Gaussian input and the performance is shown to be inferior to the ML/MAP detector, especially for higher order modulation. 

\subsection{The proposed BP algorithm over pair-wise graphical models}

In this paper, we propose using the following message passing rule. 
\begin{align}
\label{eq:a18}
\pi_{i \rightarrow j} (x_j) = \alpha \sum_{x_i \in \Xi} \tilde{p} (x_j | x_i, {\pmb y}) \prod_{k\in V(i)\backslash j} \pi_{k \rightarrow i} (x_i).
\end{align}

\noindent where $\tilde{p} (x_j | x_i, {\pmb y})$ is the conditional {\it a posteriori} probability derived under a Gaussian input assumption to be discussed shortly. Comparing with (\ref{eq:29}), the potential function in (\ref{eq:29}) is replaced with $\tilde{p} (x_j | x_i, {\pmb y})$. Note, however, that it is not a factor of the {\it a posteriori} probability in (\ref{eq:1}), unlike those in (\ref{eq:a19}).

The trick here is to use $\tilde{p} (x_j | x_i, {\pmb y})$ obtained under Gaussian input assumption in order to approximate the marginal PDF of non-Gaussian data. Note also that although the translation function $\tilde{p} (x_j | x_i, {\pmb y})$ is obtained under the Gaussian assumption on the data symbol, the message itself, $\pi_{i \rightarrow j} (x_j)$, will not be treated as Gaussian. The rationale of using $\tilde{p}(x_j|x_i, {\pmb y})$ is to reduce the computational complexity. Let ${p}(x_j|x_i, {\pmb y})$ be the true conditional {\it a posteriori} probability without the Gaussian assumption. Further assume that, after many iterations, the extrinsic information $\prod_{k \in V(i)\backslash j} \pi_{k\rightarrow i} (x_i)$ for the $i$th node (a neighbor of the $j$th node) converges to its true \emph{a posteriori} marginal distribution, $p(x_i| {\pmb y})$. Then, with an appropriate normalizing constant, we also have $\pi_{i \rightarrow j} (x_j) \rightarrow p(x_j | {\pmb y})$ for the $j$th node, which means, once converged, this translation function ensures that the final belief is given by the true marginal \emph{a posteriori} distribution. This is actually a non-sense since, before we run the algorithm, we need first to compute ${p}(x_j|x_i, {\pmb y})$, which, however, has a complexity of ML detection. Hence, at this step, we assume $x_j$s are all Gaussian to obtain $\tilde{p} (x_j | x_i, {\pmb y})$, of which the computation is much simpler as to be discussed shortly. 
%Our proposal is to use it for the message translation in (\ref{eq:29}). Note, however, that the messages, $\pi_{i \rightarrow j} (x_j)$, will not be treated as Gaussian. That is, 
It is a simple trick to use  $\tilde{p} (x_j | x_i, {\pmb y})$ obtained under Gaussian input assumption to approximate the true posterior marginal for non-Gaussian input (i.e., ${p}(x_j|{\pmb y})$). 

On the other hand, the conditional PDF, $\tilde{p} (x_j | x_i, {\pmb y})$, under Gaussian input assumption can be easily obtained from the following simple probability relations, i.e.,
\begin{align}
 \tilde{p}(x_i ,x_j |{\pmb{y}})\tilde{p}({\pmb{y}}) =  \tilde{p}({\pmb{y}}|x_i ,x_j )\tilde{p}(x_i ,x_j ) = \tilde{p}(x_j |x_i ,{\pmb{y}})\tilde{p}(x_i ,{\pmb{y}}) \nonumber
\end{align}
resulting in
\begin{align}
\tilde{p}(x_j \left| {x_i ,{\pmb{y}}} \right.) = \frac{{\tilde{p}(x_i ,x_j |{\pmb{y}} )}}{{\tilde{p}(x_i | {\pmb{y}})}}  = \frac{{\tilde{p}({\pmb{y}}|x_i ,x_j )\tilde{p}(x_j )}}{{\tilde{p}({\pmb{y}}|x_i )}} \label{eq32}
\end{align}
\noindent where
\begin{align}
\tilde{p}({\pmb{y}}|x_i ,x_j ) &= \mathcal{CN}\left( {{\pmb{y}};{\rm{ }}{\pmb{h}}_i x_i  + {\pmb{h}}_j x_j ,{\pmb{K}}_{\{ j,i\} } } \right) \label{eq33} \\
\tilde{p}({\pmb{y}}|x_i ) &= \mathcal{CN}\left( {{\pmb{y}};{\pmb{h}}_i x_i ,{\pmb{K}}_{\{ i\} } } \right)  \label{eq34} \\
\tilde{p}(x_i ) &= \mathcal{CN} \left( x_i; 0, 1\right)  \nonumber
\end{align}
with
\begin{align}
{\pmb{K}}_\Phi ^{}  &= \sigma ^2 {\pmb{I}} + \sum\nolimits_{k \notin \Phi }^{} {{\pmb{h}}_k^{} {\pmb{h}}_k^H } \label{eq35}
\end{align}
for $\Phi = \{\emph{i},\emph{j}\}$ or \{\emph{i}\}. In the second equality in (\ref{eq32}), we used the independence assumptions on $x_j$'s.

Moreover, the Gaussian input assumption leads us to a much simpler form.
First, define the conditional MMSE estimator for $x_j$ given $x_i$,
\begin{align}
{\pmb{c}}_{j|i}^{}  = {\pmb{K}}_{\{ j,i\} }^{ - 1} {\pmb{h}}_j \label{eq38}
\end{align}
and $y'_{j|i}  = {\pmb{c}}_{j|i}^H \pmb{y}$ such that
%That is, from the information theoretic optimality (sufficiency) of the MMSE estimator for Gaussian input, we have
%\begin{align}
%I(x_j ;{\pmb{y}}|x_i ) = I(x_j ;{\pmb{c}}_{j|i}^H {\pmb{y}}|x_i ) \label{eq37}
%\end{align}
%where
%\begin{align}
%{\pmb{c}}_{j|i}^{}  = {\pmb{K}}_{\{ j,i\} }^{ - 1} {\pmb{h}}_j \label{eq38}
%\end{align}
%is the conditional MMSE estimator for $\emph{x}_{j}$ given $\emph{x}_{i}$. Denoting $y'_{j|i}  = {\pmb{c}}_{j|i}^H %\pmb{y}$, 
%we have, in a similar way to (\ref{eq:15}),
\begin{align}
y'_{j|i}  = {\pmb{c}}_{j|i}^H {\pmb{y}} = a_{j|i,j} x_j  + a_{j|i,i} x_i  + n'_{j|i} \label{eq39}
\end{align}
where
\begin{align}
a_{j|i,k}  &= {\pmb{c}}_{j|i}^H {\pmb{h}}_k  = {\pmb{h}}_j^H {\pmb{K}}_{\{ j,i\} }^{ - 1} {\pmb{h}}_k^{}\hspace{8pt} \textrm{for}\; k = i\; \textrm{or}\; j \label{eq40} \\
\mathbb{E}|{n'_{j|i}}{|^2} &= {\pmb{c}}_{j|i}^H{{\pmb{K}}_{\{ j,i\} }}{\pmb{c}}_{j|i}^{} = {\pmb{h}}_j^H{\pmb{K}}_{\{ j,i\} }^{ - 1}{\pmb{h}}_j^{} \equiv \sigma _{j|i}^2. \label{eq41}
\end{align}
%The sufficiency in (\ref{eq37}) suggests using the following as the translation function,
Then, (\ref{eq32}) can be rewritten as
%instead of (\ref{eq32}), i.e.,
\begin{align}
\tilde{p}(x_j |x_i ,y'_{j|i} ) = \frac{{\tilde{p}(y'_{j|i} \left| {x_i ,x_j } \right.)\tilde{p}(x_j )}}{{\tilde{p}(y'_{j|i} \left| {x_i } \right.)}} \label{eq42}
\end{align}
with
\begin{align}
\tilde{p}(y'_{j|i} \left| {x_i ,x_j } \right.) &= \mathcal{CN}(y'_{j|i} ;a_{j|i,j} x_j  + a_{j|i,i} x_i ,\sigma _{j|i}^2 ) \label{eq43} \\
\tilde{p}(y'_{j|i} \left| {x_i } \right.) &= \mathcal{CN}(y'_{j|i};a_{j|i,i} x_i ,\sigma _{j|i}^2  + |a_{j|i,j} |^2 ). \label{eq44}
\end{align}
In (\ref{eq44}), we used $\emph{p}(\emph{x}_j ) = \mathcal{CN}(\emph{x}_j ;0,1)$.
Plugging (\ref{eq43}) and (\ref{eq44}) into (\ref{eq42}) and by replacing $p({x_j})$ with $\mathcal{CN}(x_j ;0,1)$, 
we have the simplified translation function from the derivation in the appendix.
\begin{align}
&\tilde{p}(x_j \left| {x_i ,y'_{j|i} } \right.) = \mathcal{CN}\left( {x_j ;\frac{{a_{j|i,j}^* }}{{\sigma _{j|i}^2  + |a_{j|i,j} |^2 }}\left( {y'_{j|i}  - a_{j|i,i} x_i } \right),\frac{{\sigma _{j|i}^2 }}{{\sigma _{j|i}^2  + |a_{j|i,j} |^2 }}} \right)\nonumber \\ 
& = \mathcal{CN}\left( {x_j ;\frac{1}{{1 + \sigma _{j|i}^2 }}\left( {y'_{j|i}  - a_{j|i,i} x_i } \right),\frac{1}{{1 + \sigma _{j|i}^2 }}} \right) \label{eq45}
\end{align}
where, in the last line, we used the fact that $\emph{a}_{j|i,j}$ is real valued and is equal to $\sigma^{2}_{j|i}$. 
Note that in (\ref{eq45}), the mean is the conditional MMSE estimate of $\emph{x}_{j}$ given $\emph{x}_{i}$.

Using (\ref{eq38}) to (\ref{eq45}), the proposed message passing rule can be summarized as follows.
\newline
\newline
\noindent \textsf{BP 2 over the fully-connected pair-wise graph}

\noindent\texttt{\small Given the messages in the previous iteration, $\pi_{k \rightarrow i}(x_{i})$}, 

\texttt{{\small (1) Compute the extrinsic information for all pairs $(i,j)$ with $i\neq j$ }}
\begin{align}
\lambda _{i \to j} (x_i ) = \prod\limits_{k \in V(i)\backslash j} {\pi _{k \to i} (x_i )}. \label{eqa46}
\end{align}

\texttt{{\small (2) Translate the message $\lambda _{i \to j} (x_j )$ to $\pi_{i \to j} (x_j )$ }}
\begin{align}
\pi _{i \to j} (x_j ) = \alpha \sum\limits_{x_i  \in \Xi } {\tilde{p}(x_j | x_i, {y'_{j|i}} )} \cdot \lambda _{i \to j} (x_i ). \label{eq46}
\end{align}
\texttt{{\small with ${\tilde{p}(x_j | x_i, {y'_{j|i}} )}$ given by (\ref{eq45}). The above message passing is computed for all edges 
in both directions, and they are repeated by a pre-defined number or until the \\messages do not change any more. The belief is finally obtained the same as that in (\ref{eq:30}). }}

Note that the above algorithm uses two types of message and can be efficiently described by a message passing over a bipartite graph in Fig.~\ref{Fig:4}(a), where the observations used for the message translation from the \emph{j}th node to the \emph{i}th and its reverse is clearly
denoted by $y'_{j|i}$ and $y'_{i|j}$, respectively.
It is also interesting to note that the above algorithm is similar to the algorithms in \cite{R9} and \cite{R10} with 
two differences. One is in the underlying structure and the other in message translation. To clarify the similarity 
and difference, we consider the ring-type bipartite graph shown in Fig. 2(b). In this ring-type graph, each node has only two 
neighbors and, hence, in the computation of extrinsic information, the incoming message from one neighbor is simply 
passed to the other and the detection algorithm can be described more concisely and clearly as follows 
(even though it can be generally applicable to any pair-wise graphical model). \newline
\newline
\noindent \textsf{BP 3 over the ring-type pair-wise graph (Forward-backward recursion)}

\noindent\texttt{\small Given the messages in the previous iteration, $\pi_{k \rightarrow i}(x_{i})$},
 
\noindent \texttt{{\small (1) Variable node to observation node message}}
\begin{align}
\lambda _{j \to (j \pm 1)_M} (x_{j} ) = \pi _{(j \mp 1)_M \to j} (x_j )  \; \; \; \; \; \; \forall j. \label{eq47a}
\end{align}

\noindent  \texttt{\small (2) Observation node to variable node message}
\begin{align}
\pi _{j \to (j \pm 1)_M} (x_{(j \pm 1)_M} ) = \sum \limits_{x_j  \in \Xi } {\tilde{p}(x_{(j \pm 1)_M} |x_j ,y'_{(j \pm 1)_M\; |j} ) \cdot \lambda _{j \to (j \pm 1)_M} (x_{j} ) } \; \; \; \; \; \; \forall j. \label{eq47}
\end{align}

\noindent \texttt{{\small with ${\tilde{p}(x_j | x_i, {y'_{j|i}} )}$ given by (\ref{eq45}). After a pre-defined number of iterations, the belief is finally obtained by}}
\begin{align}
b(x_j ) = \pi _{(j + 1)_M \to j} (x_j ) \cdot \pi _{(j - 1)_M \to j} (x_j ). \label{eq48}
\end{align}
\noindent From (\ref{eq47a}) to (\ref{eq48}), $(\cdot)_M$ denotes the 1-base modulo-$M$ operation such that $(M+1)_M = 1$ and $(0)_M = M$. Later on, however, we will omit this for notational simplicity.

On the other hand, this message update rule is a forward-backward algorithm similar to those in \cite{R9}, i.e., 
the message from the $(j-1)$th node to the $j$th node corresponds to the forward message, and the one from 
the $(j+1)$th node to the $j$th node corresponds to the backward message. The difference is in the message translation. 
In (\ref{eq47}), the message translation from the $j$th node to the $i$th and its reverse utilize a different 
translation function, i.e.,
\begin{align}
y'_{j|i}  \ne y'_{i|j}  \Rightarrow p(x_j |x_i ,y'_{j|i} ) = \frac{{p(y'_{j|i} \left| {x_i ,x_j } \right.)p(x_j )}}{{p(y'_{j|i} \left| {x_i } \right.)}} \ne p(x_i |x_j ,y'_{i|j} ). \nonumber
\end{align}
This means the branch metrics used for the forward and backward recursion are separately optimized to maximize their conditional SINR,
as also proposed in \cite{R10}. The translation function is also different from the branch metric in \cite{R10}, i.e., the mean
and variance in (\ref{eq45}) have a scaling factor of $a_{j|i,j}^* /(\sigma _{j|i}^2  + |a_{j|i,j} |^2 )$ and
$1/(\sigma _{j|i}^2  + |a_{j|i,j} |^2 )$, respectively, instead of $a_{j|i,j}^* /|a_{j|i,j} |^2$ and $1/|a_{j|i,j} |^2$, though
it has a minor impact on the error rate performances. The bipartite graphs corresponding to this algorithm is shown in Fig.~\ref{Fig:4}(b),
where the observation used for the message translation from the \emph{j}th node to the \emph{i}th and its reverse is clearly
denoted by $y'_{j|i}$ and $y'_{i|j}$, respectively.
Note that, for ring-type graph, we obtain different performance with a different antenna permutation, as also noted in \cite{R7},
while, in the fully-connected one, we do not need antenna permutation, which is one possible advantage of the latter
to the former.

Since the graphical models have short cycle(s) (especially the fully-connected pair-wise graph), it is quite questionable whether or not BP 2 and 3 will converge. In the literature, it was known that the convergence of BP over a loopy graph is not guaranteed, even though it does converge
in most practical cases. Since the convergence proof for non-Gaussian input is not tractable, we will tackle this question
in the next section by modifying them for Gaussian input.

\subsection{Complexity}

For complexity comparisons, we need to consider both the linear preprocessing and the post iterative detections.
Consider first the computational complexity of the post iterative detection only. In the MAP detector, the distance metric $| {\pmb y} - {\pmb{Hx}}|^2$ is computed first for all combinations of $(x_1, x_2, \cdots, x_M) \in \Xi^M$ and, then, the marginalization in (\ref{eq:4}) is performed over all combinations of ${\pmb x} \backslash x_j \in \Xi^{M-1}$ for each of $2^m$ alphabet, resulting in a complexity of $\mathcal{O}(M^2\cdot 2^{mM})$. Comparing with the complexity of MAP detector, the computational burden in the BP2 for the fully-connected pair-wise graph in Fig.~\ref{Fig:4}(a) for $\nu$ iterations is $\mathcal{O}(\nu \cdot M(M-1) \cdot 2^{2m})$ since the marginalization for each $M$ node is performed separately for its ($M-1$)  neighbors and repeated $\nu$ times. Although some additional computation is required for the linear processing in (\ref{eq35})-(\ref{eq41}), it is typically much smaller than $2^{m(M-1)}$, resulting in considerable computational reduction, which certainly comes from modeling through the pair-wise graphical model. On the other hand, the computational complexity for the ring-type pair-wise graph in Fig.~\ref{Fig:4}(b) is $\mathcal{O}(\nu\cdot M\cdot 2^{2m})$, which is even less than that of the fully-connected one.

To evaluate approximate number of operations, we assume:

\begin{enumerate}
\item The marginalization in (\ref{eq:4}) for the MAP and the computation in (\ref{eq46}) and (\ref{eq47}) for the BP 2 and 3, respectively, are performed in log-domain, where multiplications and additions in these equations are replaced with addition and max-operation, respectively, and, in (\ref{eq:3}) and (\ref{eq45}), we only need to compute its exponent.  
\item A multiplication of a $(p\times q)$ matrix with a $(q\times r)$ matrix requires $pqr$ times of multiplications and additions (of complex numbers).
\item An inversion of a $(p\times p)$ square matrix approximately requires $2p^3-2p^2$ times of additions, $2p^3-p^2$ times of multiplications and $p^2$ times of divisions (of complex numbers).
\item Division of complex numbers requires one complex multiplication and two real divisions. 
\item A complex addition requires two real additions and complex multiplication requires four real multiplications and two real additions.
\item Real addition and multiplication are assumed to have the same complexity of one (operation), while real division to have $8$ (operations).
\end{enumerate}
With these assumptions, we can count the number of operations required to generate the symbol likelihoods, i.e., the \emph{a posteriori} likelihood in (\ref{eq:4}) for the MAP and the final beliefs in the BP2 and BP3. We do not count the generation of LLR for each bit from the symbol likelihood since it is the same for all detectors. The results are summarized in Table.1, where we also show two examples, one with $M=6$, $m=2$, $\nu_1=4$, $\nu_2=6$ and the other with $M=4$, $m=4$, $\nu_1=4$, $\nu_2=6$, where $\nu_1$ and $\nu_2$ are the number of iterations for the BP2 and BP3, respectively.

It will be interesting to compare the complexity of the proposed schemes with the one in~\cite{R34} (Table I). As analyzed for BP2 and BP3, the complexity can be considered separately for the preprocessing and the post decoding. For the latter, the complexity of the one in \cite{R34} should be the same as that of the BP2,
% since they use the same graphical model, i.e., the fully-connected pair-wise graph, 
though it would be more complex than that of the BP3 in our proposal. The main difference is in the preprocessing stage. Certainly, the complexity of the preprocessing in \cite{R34} is much less than that of the proposed preprocessing since it consists of only two matrix multiplications, i.e., ${\pmb H}^H{\pmb H}$ and ${\pmb H^H}{\pmb  r}$, which requires $M^3+M^2$ of complex multiplications and the same number of complex additions.

\section{Message Passing with Gaussian Input}

In Section III, we developed BP algorithms run over the pair-wise bipartite graphs for non-Gaussian messages. The Gaussian assumption on $x_j$'s was employed first to obtain thetranslation function in (\ref{eq45}). While, we used the exact marginalization in the message translation step. In this section, we further simplify the message passing rule by extending the Gaussian assumption to the message translation step, as was done in~\cite{R16, R17, R19}, to obtain the Gaussian BP over the two graphical models under consideration. 

\emph{ML detection with Gaussian input:} With independent and identically distributed Gaussian input, $p({\pmb x}) = \prod_{j=1}^M \mathcal{CN}(x_j; 0, 1)$, the MAP detector in (\ref{eq:4}) becomes 
\begin{align}
\begin{split}
p(x_j| {\pmb y}) &= A \cdot \int  \cdots \int \mathcal{CN} ({\pmb y}; {\pmb H}{\pmb x}, \sigma^2) \prod_{k \neq j} \mathcal{CN}(x; 0,1)  d{\pmb x}\backslash x_j =\mathcal{CN} (x_j; {\pmb h}^H_j {\pmb K}^{-1} {\pmb y}, 1-{\pmb h}_j^H {\pmb K}^{-1} {\pmb h}_j) \label{eq:52}
\end{split}
\end{align}
where we appropriately select a normalization constant $A$, while the covariance matrix ${\pmb K}$, is given by ${\pmb K} = ({\pmb H} {\pmb H}^H + \sigma^2 {\pmb I})$. Noting that, in (\ref{eq:52}), the mean is the linear MMSE estimates of $x_j$ and the variance is the corresponding minimum MSE, i.e., 
\begin{align}
&\hat{x}_j = {\pmb h}_j^H {\pmb K}^{-1} {\pmb y}  \label{eq:54} \\
&\text{MMSE}_j = 1- {\pmb h}_j^H {\pmb K}^{-1} {\pmb h}_j. \label{eq:55}
\end{align}
This means that linear MMSE estimation is optimum for the Gaussian input, while it does not hold for non-Gaussian input.

\subsection{Gaussian BP over the proposed pair-wise graphs}
Assuming that $x_j$'s are Gaussian and the distributions $\pi_{i \rightarrow j} (x_j)$, and $b(x_i)$ are all Gaussian PDFs, they can be characterized by their mean and variance only. This means the messages $\pi_{i \rightarrow j} (x_j)$ and the belief, $b(x_i)$, in the BP 2 and 3 can be replaced with the update rule for the mean and variance pair. Since the Gaussian BP corresponding to the BP 1 over the fully connected pair-wise graph in (\ref{eq:9})-(\ref{eq:11}) has already been discussed in~\cite{R16}, we consider here only the BP 2 and 3 over the two pair-wise graphical models. 

Let us denote the mean and the variance pair of the complex Gaussian PDFs, $\pi_{i \rightarrow j} (x_j)$, and $b(x_i)$ as $(\mu_{\pi, i\rightarrow j}, \sigma^2_{\pi, i\rightarrow j})$ and $(\mu_{i}, \sigma^2_{i})$. Then, the BP 2 and 3 under the Gaussian input assumption can be rewritten as follows (Detailed derivations are shown in the appendix): \newline
\newline
\noindent \textsf{Gaussian BP 2G over the fully-connected pair-wise graph}

\noindent \texttt{\small
Given the messages in the previous iteration (or the initial messages),
$\pi_{i\rightarrow j}(x_i)=(\mu_{\pi,i\rightarrow j}, \sigma^{2}_{\pi,i\rightarrow j})$ $\;\forall (i,j): i \ne j$, 
they are recursively updated by}
\begin{align}
\sigma _{\pi ,i \to j}^2  &= \frac{1}{{1 + \sigma _{j|i}^2 }} + \frac{{|a_{j|i,i} |^2 }}{{(1 + \sigma _{j|i}^2 )^2 }} \cdot \sum\nolimits_{k \in V(i)\backslash j} {\sigma _{\pi ,k \to i}^{ - 2} } \label{eq58} \\
\mu _{\pi ,i \to j}  &= \frac{{y'_{j|i} }}{{1 + \sigma _{j|i}^2 }} - \frac{{a_{j|i,i} }}{{1 + \sigma _{j|i}^2 }} \cdot \frac{{\sum\nolimits_{k \in V(i)\backslash j} {\sigma _{\pi ,k \to i}^{ - 2} \mu _{\pi ,k \to i}^{} } }}{{\sum\nolimits_{k \in V(i)\backslash j} {\sigma _{\pi ,k \to i}^{ - 2} } }}. \label{eq59}
\end{align}
\texttt{\small After a number of iterations of the above, the final belief on $\emph{x}_{i}$ is obtained by}
\begin{align}
\sigma _i^{ - 2}  &= \sum\nolimits_{k \in V(i)} {\sigma _{\pi ,k \to i}^{ - 2} }  \label{eq60} \\
\mu _i  &= \frac{{\sum\nolimits_{k \in V(i)} {\sigma _{\pi ,k \to i}^{ - 2} \mu _{\pi ,k \to i}^{} } }}{{\sum\nolimits_{k \in V(i)} {\sigma _{\pi ,k \to i}^{ - 2} } }}. \label{eq61}
\end{align}

\noindent \textsf{Gaussian BP 3G over the ring-type MRF (Gaussian forward-backward recursion)}

\noindent\texttt{\small Given the messages in the previous iteration, 
$\pi_{i \rightarrow i \pm 1}(\emph{x}_{i})=(\mu_{\pi,i\rightarrow i \pm1},\sigma^{2}_{\pi,i\rightarrow i \pm 1})$ $\;\forall i$,
they are recursively updated by}
\begin{align}
\sigma _{\pi ,i \to i \pm 1}^2  &= \frac{1}{{1 + \sigma _{i \pm 1|i}^2 }} + \frac{{|a_{i \pm 1|i,i} |^2 }}{{(1 + \sigma _{i \pm 1|i}^2 )^2 }} \cdot \sigma _{\pi ,i \mp 1 \to i}^{ - 2} \label{eq64} \\
\mu _{\pi ,i \to i \pm 1}  &= \frac{1}{{1 + \sigma _{i \pm 1|i}^2 }}y'_{i \pm 1|i}  - \frac{{a_{i \pm 1|i,i} }}{{1 + \sigma _{i \pm 1|i}^2 }} \cdot \mu _{\pi ,i \mp 1 \to i}^{}. \label{eq65}
\end{align}
\texttt{\small After a number of iterations of the above, the final belief on $\emph{x}_{i}$ is obtained by}
\begin{align}
\sigma _i^{ - 2}  &= \sigma _{\pi ,i + 1 \to i}^{ - 2}  + \sigma _{\pi ,i - 1 \to i}^{ - 2}  \label{eq66} \\
\mu _i  &= \frac{{\sigma _{\pi ,i + 1 \to i}^{ - 2} \mu _{\pi ,i + 1 \to i}^{}  + \sigma _{\pi ,i - 1 \to i}^{ - 2} \mu _{\pi ,i + 1 \to i}^{} }}{{\sigma _{\pi ,i + 1 \to i}^{ - 2}  + \sigma _{\pi ,i - 1 \to i}^{ - 2} }}.  \label{eq67}
\end{align}
\newline
Particularly, in the Gaussian BP 3G, we observe the following:

\begin{enumerate}
  \item The variance and mean are updated separately (except in the final belief).
  \item In (\ref{eq64}) and (\ref{eq65}), there are two separate message flows; one is the forward from \emph{i} to \emph{i}+1 and the other is the backward from \emph{i} to \emph{i}-1.
  \item Eq. (\ref{eq65}) can be rewritten as
\begin{align}
\textrm{Forward recursion:}\; \mu _{\pi ,i \to i + 1}  = F_i  \circ \mu _{\pi ,i - 1 \to i}^{} \label{eq68} \\
\textrm{Backward recursion:}\; \mu _{\pi ,i \to i - 1}  = B_i  \circ \mu _{\pi ,i + 1 \to i}^{} \label{eq69}
\end{align}
where the operations, $\emph{F}_{i}$ and $\emph{B}_{i}$, are first order elementary function defined as
\begin{align}
F_i  \circ \mu  \equiv u_{i + 1,i}  + v_{i + 1,i}  \cdot \mu \label{eq70} \\
B_i  \circ \mu  \equiv u_{i - 1,i}  + v_{i - 1,i}  \cdot \mu \label{eq71}
\end{align}
with
\begin{align}
u_{j,i}&=\frac{{y'_{j|i} }}{{1 + \sigma _{j|i}^2 }} = {\rm{  }}\frac{{{\pmb{h}}_j^H {\pmb{K}}_{\{ j,i\} }^{ - 1} {\pmb{y}}}}{{1 + {\pmb{h}}_j^H {\pmb{K}}_{\{ j,i\} }^{ - 1} {\pmb{h}}_j^{} }}{\rm{ }} = {\pmb{h}}_j^H {\pmb{K}}_{\{ i\} }^{ - 1} {\pmb{y}} \label{eq72}
\end{align}
\begin{align}
v_{j,i}&=\frac{{-a_{j|i,i} }}{{1 + \sigma _{j|i}^2 }} =  \frac{{{-\pmb{h}}_j^H {\pmb{K}}_{\{ j,i\} }^{ - 1} {\pmb{h}}_i^{} }}{{1 + {\pmb{h}}_j^H {\pmb{K}}_{\{ j,i\} }^{ - 1} {\pmb{h}}_j^{} }} =  - {\pmb{h}}_j^H {\pmb{K}}_{\{ i\} }^{ - 1} {\pmb{h}}_i^{} .\label{eq73}
\end{align}
Here, we used (\ref{eq38})-(\ref{eq41}) and, in the last, the matrix inversion lemma
\begin{align}
({\pmb{A}} + {\pmb{BB}}^H )^{ - 1}  = {\pmb{A}}^{ - 1}  - {\pmb{A}}^{ - 1} {\pmb{B}}({\pmb{I}} + {\pmb{B}}^H {\pmb{A}}^{ - 1} {\pmb{B}})^{ - 1} {\pmb{B}}^H {\pmb{A}}^{ - 1}. \nonumber
\end{align}
  \item Similar to the means, (\ref{eq64}) can also be rewritten as
\begin{align}
\textrm{Forward recursion:}\; \sigma _{\pi ,i \to i + 1}^2  = F'_i  \circ \sigma _{\pi ,i - 1 \to i}^2  \label{eq74} \\
\textrm{Backward recursion:}\; \sigma _{\pi ,i \to i - 1}^2  = B'_i  \circ \sigma _{\pi ,i + 1 \to i}^2 \label{eq75}
\end{align}
where
\begin{align}
{F'_i} \circ \mu  \equiv {u'_{i + 1,i}} + {v'_{i + 1,i}} \cdot \mu  \label{eq76} \\
{B'_i} \circ \mu  \equiv {u'_{i - 1,i}} + {v'_{i - 1,i}} \cdot \mu  \label{eq77}
\end{align}
with
\begin{align}
u'_{j,i}  &= \frac{1}{{1 + \sigma _{j|i}^2 }} = {\rm{  }}\frac{1}{{1 + {\pmb{h}}_j^H {\pmb{K}}_{\{ j,i\} }^{ - 1} {\pmb{h}}_j^{} }}  \label{eq78} \\
v'_{j,i}  &= \frac{{|a_{j|i,i} |^2 }}{{(1 + \sigma _{j|i}^2 )^2 }} = \left| {\frac{{{\pmb{h}}_j^H {\pmb{K}}_{\{ j,i\} }^{ - 1} {\pmb{h}}_i^{} }}{{1 + {\pmb{h}}_j^H {\pmb{K}}_{\{ j,i\} }^{ - 1} {\pmb{h}}_j^{} }}} \right|^2  = \left| {{\pmb{h}}_j^H {\pmb{K}}_{\{ i\} }^{ - 1} {\pmb{h}}_i^{} } \right|^2.  \label{eq79} 
\end{align}
\end{enumerate}

\subsection{Convergence of Gaussian BP}
Regarding the convergence of Gaussian BP, it was previously shown in~\cite{R24} that Gaussian BP for arbitrary topology converges to the correct mean (see also~\cite{R29}). It was shown in~\cite{R16} that the Gaussian BP over the factor graph in Fig.~1(b) converges to the linear MMSE solution, even though its convergence is not assured. Based on these findings, we can conjecture that, for both the Gaussian BP of rules 2G and 3G, the mean converges to the linear MMSE solution, as also verified by simulations in the next section. One way to prove the convergence would be to use the idea of the ``unwrapped tree" presented in~\cite{R24}. In our case, however, this would be a tedious derivation. Therefore, we try an alternative approach that works for GBP 3G, but not for GBP 2G. Note, however, that the derivation here differs from~\cite{R16, R17} in the underlying graphical model and the translation function used. The objective in this subsection is to prove the following theorem.

\begin{theorem}
In the Gaussian BP 3G over the ring-type pair-wise graph, the mean converges to the linear MMSE estimate (\ref{eq:54}) for non-zero noise power as the number of iterations approaches infinity.
\end{theorem}

The proof is based on the following \emph{Lemmas}. 

\begin{lemma}
For an arbitrary initial value $\mu(0)$, both the forward and backward recursions for the mean in (\ref{eq65}) converge respectively to a unique, fixed point.
\end{lemma}

\begin{proof}
Define one iteration as one complete turn of a message passing along the ring and consider, without loss of generality, the message at Node 1. Based on observations 1) through 3) in the previous subsection, we obtain the recursive relations for Node 1, i.e., using an arbitrary initial value $\mu(0)$, we have
\begin{align}
\begin{split}
\label{eq:80} \mu_{\pi, 1 \rightarrow 2} (n) =& \left(  F_M \circ \cdots F_3 \circ F_2 \circ F_1\circ \right) \mu_{\pi, 1\rightarrow 2} (n-1)  \\
=&(F_M \circ \cdots F_3 \circ F_2 \circ F_1\circ)^n \mu(0)
\end{split}\\
\begin{split}
\label{eq:81} \mu_{\pi,1 \rightarrow M} (n) =& \left(B_2 \circ B_3 \circ \cdots B_M \circ B_1 \circ \right) \mu_{\pi,1 \rightarrow M} (n-1)\\
=&  \left(B_2 \circ B_3 \circ \cdots B_M \circ B_1 \circ \right)^n \mu(0)
\end{split}
\end{align}
where $n$ is the iteration number and the collective operations for one iteration of the forward/backward recursion are given, respectively, by
\begin{align}
\label{Eq:F1} F_{1,T} \circ \mu =& F_M \circ \cdots F_3 \circ F_2 \circ F_1 \circ \mu = f_{1,U} + f_{1,V} \mu\\\
\label{Eq:B1} B_{1,T} \circ \mu =& B_2 \circ B_3 \circ \cdots B_M \circ B_1 \circ \mu = b_{1,U} + b_{1,V} \mu
\end{align}
for some constants, $ f_{1,U}, f_{1,V}, b_{1,U}, $and $b_{1,V}$, which, in turn, are monomials of $u_{j,i}$ and $v_{j,i}$ in (\ref{eq70}) and (\ref{eq71}). For example, we have for $M=4$
\begin{align}
F_4 \circ F_3 \circ F_2 \circ F_1 \circ \mu =& \left(u_{1,4}+v_{1,4}u_{4,3}+v_{1,4}v_{4,3}u_{3,2}+v_{1,4}v_{4,3}v_{3,2}u_{2,1}\right)+\left(v_{1,4}v_{4,3}v_{3,2}v_{2,1}\right)\cdot\mu \nonumber \\
B_2 \circ B_3 \circ B_4 \circ B_1 \circ \mu =& \left(u_{1,2}+v_{1,2}u_{2,3}+v_{1,2}v_{2,3}u_{3,4}+v_{1,2}v_{2,3}v_{3,4}u_{4,1}\right)+\left(v_{1,2}v_{2,3}v_{3,4}v_{4,1}\right)\cdot\mu \nonumber
\end{align} 
for which
\begin{align}
\nonumber f_{1,U} =& u_{1,4}+v_{1,4}u_{4,3}+v_{1,4}v_{4,3}u_{3,2}+v_{1,4}v_{4,3}u_{3,2}u_{2,1}\\
\nonumber f_{1,V} =& v_{1,4}v_{4,3}v_{3,2}v_{2,1}\\
\nonumber b_{1,U} =& u_{1,2}+v_{1,2}u_{2,3}+v_{1,2}v_{2,3}u_{3,4}+v_{1,2}v_{2,3}v_{3,4}u_{4,1}\\
\nonumber b_{1,V} =& v_{1,2}v_{2,3}v_{3,4}v_{4,1}.
\end{align}
Here, we can show that $f_{1,V}$ and $b_{1,V}$ are given, respectively, by
\begin{align}
f_{1,V} = \prod_{j=1}^M v_{j,j-1} ~~\text{and}~~ b_{1,V} = \prod_{j=1}^M v_{j,j+1}. \label{eq:84}
\end{align}
On the other hand, using (\ref{Eq:F1}) and (\ref{Eq:B1}), (\ref{eq:80}) and (\ref{eq:81}) become
\begin{align}
\mu_{\pi,1\rightarrow2}(n) =& (F_{1,T}\circ)^n\mu(0) = f_{1,U}\cdot\sum_{k=0}^{n-1}f_{1,V}^k + f_{1,V}^n\cdot\mu(0) \label{eq:85} \\
\mu_{\pi,1\rightarrow M}(n) =& (B_{1,T}\circ)^n\mu(0) = b_{1,U}\cdot\sum_{k=0}^{n-1}b_{1,V}^k + b_{1,V}^n\cdot\mu(0) \label{eq:86}
\end{align}
where, from the fact to be proved in the next \emph{Lemma} that $|f_{i,V}|<1$ and $|b_{i,V}|<1$, we have
\begin{align}
f_{1,V}^n\cdot\mu(0)\rightarrow0,~b_{1,V}^n\cdot\mu(0)\rightarrow0~\text{as}~n\rightarrow\infty. \nonumber
\end{align}
Therefore, the unique fixed point of the mean in GBP 3G is given by
\begin{align}
\lim_{n \rightarrow \infty}\mu_{\pi,1\rightarrow2}(n)~\rightarrow~f_{1,U} \cdot\sum_{k=0}^\infty f_{1,V}^k=\frac{f_{1,U}}{1-f_{1,V}} \label{eq:87} \\
\lim_{n \rightarrow \infty}\mu_{\pi,1\rightarrow M}(n)~\rightarrow~b_{1,U} \cdot\sum_{k=0}^\infty b_{1,V}^k=\frac{b_{1,U}}{1-b_{1,V}}. \label{eq:88}
\end{align}
\end{proof}

\begin{lemma}
$|f_{i,V}|=|\prod_{j=1}^M v_{j,j-1}|<1~\text{and}~|b_{i,V}|=|\prod_{j=1}^M v_{j,j+1}|<1~$for all $i$.
\end{lemma}

\begin{proof}
By plugging  into (\ref{eq73}) into (\ref{eq:84}), we have for all $i$.
%\begin{align}
%\nonumber |f_{i,V}| &= \left|  \prod_{j=1}^M \frac {{\pmb h}_j^H {\pmb K}_{\{j,j-1 \}}^{-1} {\pmb h}_{j-1}} {1+{\pmb h}_j^H {\pmb K}_{\{j,j-1 \}}^{-1} {\pmb h}_{j}} \right| \\
%\nonumber &= \left| \prod_{j=1}^M {\pmb h}_j^H {\pmb K}_{\{j-1\}}^{-1} {\pmb h}_{j-1} \right| \\
%\nonumber &= \left| {\pmb h}_1^H {\pmb K}_{M}^{-1} {\pmb h}_M {\pmb h}_M^H {\pmb K}_{\{M-1\}}^{-1} \cdots  {\pmb K}_{2}^{-1} {\pmb h}_2 {\pmb h}_2^H {\pmb K}_{1}^{-1} {\pmb h}_1  \right| \\
%\nonumber &= \left| \text{tr}\left( {\pmb h}_1 {\pmb h}_1^H {\pmb K}_{\{M\}}^{-1} {\pmb h}_M {\pmb h}_M^H {\pmb K}_{\{M-1\}}^{-1} \cdots {\pmb K}_{2}^{-1} {\pmb h}_2 {\pmb h}_2^H {\pmb K}_{1}^{-1} \right) \right| \\
%\nonumber & \leq \left| \prod_{j=1}^M \text{tr}({\pmb h}_j {\pmb h}_j^M {\pmb K}_{\{j-1\}}^{-1}) \right| \\
%&= \left| \prod_{j=1}^M {\pmb h}_j^H {\pmb K}_{\{j-1\}}^{-1} {\pmb h}_j \right| \nonumber \\
%&= \prod_{j=1}^M \frac{{\pmb h}_j^H {\pmb K}_{\{j,j-1\}}^{-1} {\pmb h}_j}{1+{\pmb h}_j^H {\pmb K}_{\{j,j-1\}}^{-1} {\pmb h}_j} \nonumber\\
%&= \prod_{j=1}^M \frac{\sigma_{j|j-1}^2}{1+\sigma_{j|j-1}^2}<1 \nonumber
%\end{align} 
\begin{align}
\nonumber |f_{i,V}| &= \left|  \prod_{j=1}^M \frac {{\pmb h}_j^H {\pmb K}_{\{j,j-1 \}}^{-1} {\pmb h}_{j-1}} {1+{\pmb h}_j^H {\pmb K}_{\{j,j-1 \}}^{-1} {\pmb h}_{j}} \right| = \left| \prod_{j=1}^M {\pmb h}_j^H {\pmb K}_{\{j-1\}}^{-1} {\pmb h}_{j-1} \right| \\
\nonumber &= \left| {\pmb h}_1^H {\pmb K}_{M}^{-1} {\pmb h}_M {\pmb h}_M^H {\pmb K}_{\{M-1\}}^{-1} \cdots  {\pmb K}_{2}^{-1} {\pmb h}_2 {\pmb h}_2^H {\pmb K}_{1}^{-1} {\pmb h}_1  \right| \\
\nonumber &\overset{(a)}= \left| \text{tr}\left( {\pmb h}_1 {\pmb h}_1^H {\pmb K}_{\{M\}}^{-1} {\pmb h}_M {\pmb h}_M^H {\pmb K}_{\{M-1\}}^{-1} \cdots {\pmb K}_{2}^{-1} {\pmb h}_2 {\pmb h}_2^H {\pmb K}_{1}^{-1} \right) \right| \\
\nonumber & \overset{(b)}\leq \left| \prod_{j=1}^M \text{tr}({\pmb h}_j {\pmb h}_j^M {\pmb K}_{\{j-1\}}^{-1}) \right| = \left| \prod_{j=1}^M {\pmb h}_j^H {\pmb K}_{\{j-1\}}^{-1} {\pmb h}_j \right| \nonumber \\
&\overset{(c)}= \prod_{j=1}^M \frac{{\pmb h}_j^H {\pmb K}_{\{j,j-1\}}^{-1} {\pmb h}_j}{1+{\pmb h}_j^H {\pmb K}_{\{j,j-1\}}^{-1} {\pmb h}_j} \nonumber= \prod_{j=1}^M \frac{\sigma_{j|j-1}^2}{1+\sigma_{j|j-1}^2}<1 \nonumber
\end{align} 
where, $(a)$ follows by the fact that ${\pmb a}^H {\pmb b} = \text{tr}({\pmb b} {\pmb a}^H)$ for arbitrary vectors ${\pmb a}$ and ${\pmb b}$, and $(b)$ results from $\text{tr}({\pmb A} {\pmb B}) \leq \text{tr}({\pmb A})\text{tr}({\pmb B})$ for arbitrary non-negative definite matrices ${\pmb A}$ and ${\pmb B}$. Also, $( c)$ follows by the following matrix inversion \emph{Lemma}
\begin{align}
{\pmb h}_j^H {\pmb K}_{\{j+1\}}^{-1} &= {\pmb h}_j^H\left({\pmb K}_{\{j,j+1\}}+{\pmb h}_j{\pmb h}_j^H \right)^{-1} ={\pmb h}_j^H \left( {\pmb K}_{\{j,j+1\}}^{-1}-{\pmb K}_{\{j,j+1\}}^{-1} {\pmb h}_j (1+{\pmb h}_j^H {\pmb K}_{\{j,j+1\}}^{-1} {\pmb h}_j)^{-1} {\pmb h}_j^H {\pmb K}_{\{j,j+1\}}^{-1} \right) \nonumber\\
&= \left(1-\frac{{\pmb h}_j^H {\pmb K}_{\{j,j+1\}}^{-1} {\pmb h}_j}{1+{\pmb h}_j^H {\pmb K}_{\{j,j+1\}}^{-1} {\pmb h}_j}\right) {\pmb h}_j^H {\pmb K}_{\{j,j+1\}}^{-1}   = \left(\frac{1}{1+{\pmb h}_j^H {\pmb K}_{\{j,j+1\}}^{-1} {\pmb h}_j}\right) {\pmb h}_j^H {\pmb K}_{\{j,j+1\}}^{-1}.\nonumber
\end{align}
For the backward recursion, $|b_{i,V}|<1$ can also be proved in a similar way.
\end{proof}

%%%%

In (\ref{eq:85}) and (\ref{eq:86}), we see that the convergence rate depends on
\begin{align}
\left| \prod_{j=1}^M {\pmb h}_j^H {\pmb K}^{-1}_{\{j-1  \}} {\pmb h}_{j-1}    \right| \leq \prod_{j=1}^M \frac{\sigma^2_{j|j-1}}{1+\sigma^2_{j|j-1}} <1 \nonumber
\end{align} 
which is similar to the result in\cite{R16}. Note that ${\pmb h}_i^H {\pmb K}^{-1}_{\{i-1  \}} {\pmb h}_{i-1}$ reflects the channel correlation between neighboring antennas.

On the other hand, the operations, $F_i$ and $B_i$, are not permutable, such that $F_i \circ F_j \circ \mu$ and $F_j \circ F_i \circ \mu$ may be different, and so are $F_{j, T} \circ \mu$ and $F_{i,T} \circ \mu$ for $j \neq i$. That is, the fixed point for each node may differ from one another. 

The following two \emph{Lemmas} show that the fixed points in (\ref{eq:87}) and (\ref{eq:88}) are both equal to the MMSE estimate in (\ref{eq:54}). 
\begin{lemma}
In the forward recursion, $\mu_{\pi, i \rightarrow i+1} (n)$  is the linear MMSE estimates of $x_{i+1}$ provided that the previous message, $\mu_{\pi, i-1 \rightarrow i}(n)$, is the linear MMSE estimates of $x_i$. Likewise, in the backward recursion, $\mu_{\pi, i \rightarrow i-1}(n) $, is the linear MMSE estimates of $x_{i-1}$ provided that $\mu_{\pi, i+1 \rightarrow i}(n)$ is the linear MMSE estimates of $x_i$.
\end{lemma}

\begin{proof}
With ${\pmb c}_j = {\pmb K}^{-1} {\pmb h}_j$, the linear MMSE estimate of $x_i$ is given by ${\pmb h}_i^H {\pmb K}^{-1} {\pmb y}$. And, hence, the proof is to show from (\ref{eq68}) and (\ref{eq69}) that
\begin{align}
{\pmb h}^H_{i+1} {\pmb K}^{-1} {\pmb y} &= F_i \circ ({\pmb h}_i^H {\pmb K}^{-1} {\pmb y}) = u_{i+1, i} + \nu_{i+1, i} \cdot ({\pmb h}^H_i {\pmb K}^{-1} {\pmb y})\label{eq:89} \\
{\pmb h}^H_{i-1} {\pmb K}^{-1} {\pmb y} &= B_i \circ({\pmb h}_i^H {\pmb K}^{-1} {\pmb y}) = u_{i-1, i} + \nu_{i-1, i} \cdot ({\pmb h}^H_i {\pmb K}^{-1} {\pmb y}) \nonumber
\end{align}
where $u_{j,i}$ and $\nu_{j,i}$ are given by (\ref{eq72}) and (\ref{eq73}). Plugging these into the right hand side of (\ref{eq:89}) for the forward recursion, we finally have
\begin{align}
\begin{split}
u_{i+1,i} + \nu_{i+1, i} \cdot ({\pmb h}_i^H {\pmb K}^{-1} {\pmb y}) &= {\pmb h}_{i+1}^H {\pmb K}^{-1}_{\{i\}} {\pmb y} -  {\pmb h}_{i+1}^H {\pmb K}^{-1}_{\{i\}} {\pmb h}_i {\pmb h}_i^H {\pmb K}^{-1} {\pmb y} = {\pmb h}^H_{i+1} \left( {\pmb K}^{-1}_{\{i\}}  - {\pmb K}^{-1}_{\{i\}}   {\pmb h}_i {\pmb h}_i^H {\pmb K}^{-1}    \right) {\pmb y}\\
&=  {\pmb h}^H_{i+1} \left( {\pmb K}^{-1}_{\{i\}}  - {\pmb K}^{-1}_{\{i\}}   {\pmb h}_i \frac{{\pmb h}_i^H {\pmb K}^{-1}_{\{i \}}}{1+{\pmb h}_i^H } {\pmb K}^{-1}_{\{i \}} {\pmb h}_i  \right) {\pmb y}={\pmb h}^H_{i+1} {\pmb K}^{-1} {\pmb y}. \nonumber
\end{split}
\end{align}
Similarly, for the backward recursion, we obtain
\begin{align}
u_{i-1,i} + \nu_{i-1,i} \cdot ({\pmb h}^H_i {\pmb K}^{-1} {\pmb y}) = {\pmb h}^H_{i-1} {\pmb K}^{-1} {\pmb y}. \nonumber
\end{align}
\end{proof}

\begin{lemma}
Both the fixed points, $\frac{f_{j,U}}{1-f_{j,V}}$ and $\frac{b_{j,U}}{1-b_{j,V}}$, are equal to the MMSE estimate of $x_j$, i.e., ${\pmb h}^H_{j} {\pmb K}^{-1} {\pmb y}$.
\end{lemma}

\begin{proof}
Without loss of generality, let us consider the first data symbol, $x_1$. Starting from ${\pmb h}^H_{2} {\pmb K}^{-1} {\pmb y} = F_1 \circ ({\pmb h}^H_{1} {\pmb K}^{-1} {\pmb y})$, we can successively apply the operations, $F_2 \circ$, $F_3 \circ$,..., $F_M \circ$ to finally obtain
\begin{align}
F_M \circ F_{M-1} \circ \cdot \cdot \cdot \circ F_2 \circ F_1 \circ ({\pmb h}^H_{1} {\pmb K}^{-1} {\pmb y}) &=  F_{1,T} \circ ({\pmb h}^H_{1} {\pmb K}^{-1} {\pmb y}) \nonumber \\
&=  f_{1,U} + f_{1,V} \cdot ({\pmb h}^H_{1} {\pmb K}^{-1} {\pmb y}) \nonumber \\
&= ({\pmb h}^H_{1} {\pmb K}^{-1} {\pmb y}) \nonumber
\end{align} 
\noindent where the first and second equality are obtained by the definition of $F_{1,T} \circ$ and $( f_{1,U}, f_{1,V})$, respectively, and the last equality is from \emph{Lemma 4}, i.e., ${\pmb h}^H_{1} {\pmb K}^{-1} {\pmb y} = F_M \circ ({\pmb h}^H_{M} {\pmb K}^{-1} {\pmb y})$. From the last equality, we obtain $\frac{f_{j,U}}{1-f_{j,V}} = {\pmb h}^H_{j} {\pmb K}^{-1} {\pmb y}$ and, in a similar way, we also can prove that $\frac{b_{j,U}}{1-b_{j,V}} = {\pmb h}^H_{j} {\pmb K}^{-1} {\pmb y}$.
\end{proof}

\begin{proof}
The proof of \emph{Theorem~1} is now obvious from the above lemmas, i.e., from \emph{Lemma 2} the mean in the Gaussian BP over the ring-type graphical model converges to a unique fixed point and \emph{Lemma 5} shows that the fixed point of the mean is equal to the linear MMSE estimates in (\ref{eq:54}).
\end{proof}

Note that the $Theorem$~1 holds for any channel matrices if noise variance is not zero since, for $\sigma^2 > 0$, the covariance matrices, ${\pmb K}_{\{i,j\}}$'s in (\ref{eq35}) are always invertible so that there certainly exist the MMSE estimator in (\ref{eq38}) and the translation functions in (\ref{eq45}) for all pair of $(i,j)$.

Since the message-update rule for the variance in (\ref{eq74}) and (\ref{eq75}) have the same form as in (\ref{eq68}) and (\ref{eq69}), we can also prove the convergence of the variance in GBP 3G, which can be summarized by the following \emph{Lemma}.

\begin{lemma}
For an arbitrary initial value $\sigma^2(0)$, both the forward and backward recursions for the variance in (\ref{eq58}) and (\ref{eq60}) converge, respectively, to a unique fixed point.
\end{lemma}
The proof is similar to that of \emph{Lemma~3}. Unfortunately, however, the fixed point is not necessarily correct. That is, it may not equal the MMSE in (\ref{eq:55}), as also confirmed in~\cite{R24}. In~\cite{R28}, the convergence property of the BP over such ring-type graph was shown to be optimal for binary input. For Gaussian input, however, it is optimal only in the mean, i.e., the mean converges to the fixed point that is equal to the MMSE estimates, $\hat{x}_j$ in (\ref{eq:54}), and we cannot say so in a strict sense since the fixed point of the variance is not equal to the MSE in (\ref{eq:55}), the MAP estimates on the variance. 

It will also be worth comparing GBP 2G and 3G proposed in this paper and the Gaussian BP in \cite{R17}, \cite{R19}, and\cite{R16}, all of which are based on the direct decomposition of Gaussian PDF, and, as noticed in \cite{R19}, are the same algorithm. The comparison can be made in several aspects, i.e., in complexity and convergence. In complexity, the Gaussian BP in \cite{R17}, \cite{R19}, and \cite{R16} is much simpler than GBP 2G and 3G proposed here. Note that (1) Gaussian BP in \cite{R17}, \cite{R19}, and \cite{R16} does not require preprocessing while GBP 2G and 3G in this paper do and (2) the complexity of the post iteration for the former is obviously the same as that of GBP 2G since they utilize the same graphical model, even though the post iteration of GBP 3G is a little bit less complex than GBP 2G. Based on this, the overall complexity of the proposed GBP 2G and 3G is certainly more complex than those in \cite{R17}, \cite{R19}, and \cite{R16}. Now, let us consider their convergence. Basically, GBP 3G proposed in this paper and the Gaussian BP in \cite{R17}, \cite{R19}, and \cite{R16} results in an MMSE solution (in mean) if they converge, as proved here for GBP 3G and in \cite{R16} for Gaussian BP with the direct decomposition. This means that, once converged, they will perform the same. Unfortunately, the convergence of the Gaussian BP in \cite{R17}, \cite{R19}, and \cite{R16} seems not to be assured while GBP 3G surely converges.

%%%%=====
\section{Simulation Results}
In this section, we present simulation results for the iterative algorithms with and without channel coding. For channel coding, we used DVB-S2 LDPC code of rates 3/4 and length 64800 ~\cite{R30}. The performances of ML, MMSE, and the bi-diagonalization approach in~\cite{R9} are also evaluated as references. In the transmitter, a block (48600 bits) of random information bits are generated first and then coded using the LDPC encoder and then interleaved with a random interleaver and modulated into a sequence of $2^m$-ary symbols. The symbol sequence is then divided into sub-blocks of $M$ symbols, each of which is fed to a transmit antenna, where $M$ corresponds to the number of transmit antennas. At the receiver, the sequence of received vectors is passed to a MIMO detector, which generates the estimates of symbol likelihoods and LLRs for each coded bit . The LLR is then de-interleaved and decoded by using a generic LDPC decoder\footnote{In the transmitter and the receiver, the interleaving/de-interleaving and channel coding/decoding is used if channel coding is applied.} Note that no 'turbo principle' is applied since it is not our focus in this paper. This means that the LDPC decoding begins only after the inner iteration in MIMO detector is finished. Regarding the MIMO channel, we generated, for each transmitted data vector, an independent and identically distributed (i.i.d.) MIMO channel matrix, of which each element is also an i.i.d. complex Gaussian random variable with mean 0 and variance 1. The resulting channel can be regarded as a fully interleaved frequency selective MIMO channel that can be seen on top of the orthogonal frequency division multiplexing (OFDM), especially for those channels where the transmission bandwidth is much larger than the channel coherence bandwidth.

Fig.~\ref{Fig:5} shows a comparison of bit error rate performance as a function of signal-to-noise ratio (SNR) $(1/\sigma^2)$ for ML, BP1 in~\cite{R18}, MMSE, the bi-diagonalization approach in~\cite{R9}, and the proposed BP-based detector with the fully-connected and ring type pair-wise model. We use a $4\times 4$ antenna configuration and QPSK modulation. We could confirm from Fig.~\ref{Fig:5} that the pair-wise MRF-based detector performs as well as the ML with soft decisions (i.e., using (\ref{eq:4}) and (\ref{eq:5})). The SNR gap between the proposed scheme and the ML is shown to be around 0.1 and 0.3, respectively. 

Fig. \ref{Fig:8a} shows a comparison of bit error rate performance without channel coding as a function of signal-to-noise ratio (SNR) (1/$\sigma^2$) for ML, MMSE, and the proposed detector of fully-connected and ring type model. We used the same antenna configuration and modulation size, but without channel coding. As shown in the figure, the tendency in the relative performance looks similar to that with channel coding.

It is also worth comparing the performance of BP2 and BP3 with the one in \cite{R34} (Table-I), where the pair-wise MRF obtained by the direct decomposition of Gaussian PDF is used. The performance comparison is shown in Fig.~\ref{Fig:r2} for BPSK modulation and the same LDPC coding. As shown in the figure, the performance of the one in \cite{R34} is almost the same as that of BP2. We set the number of iterations to 4 for BP2 and BP3 and 6 for the one in \cite{R34}. We also tried to obtain the results for 4PAM. Unfortunately, however, the algorithm in \cite{R34} failed to work for 4-PAM and this is one of the advantages of using the proposed scheme over the existing (fully-connected) MRF based MIMO detection.\footnote{ The reason we consider here only one-dimensional constellation like BPSK or 4-PAM is that the algorithm in \cite{R34} is applicable only to those real constellation and we just wanted to use the algorithm as is since any modification may cause unexpected results.} 

Fig.~\ref{Fig:6} shows the BER performance for a $6 \times 6$ antenna configuration with the same modulation and channel coding. The SNR gap between the proposed scheme and the ML is now approximately 0.75 dB for the fully connected pair-wise graph and 1 dB for the ring-type one, respectively. Although the performance degradation compared to the ML is larger than for a $4\times 4$ antenna configuration, the SNR gain over the MMSE detector is around 3.5 dB. 

In Fig.~\ref{Fig:8}, BER performance with higher modulation order (16-QAM) is shown. We used a $4\times 4$ antenna configuration and the same channel coding. Here, the SNR gap between the proposed method and the ML is shown to be around 1 dB for the BP2 over the fully connected pair-wise graph and 0.7 dB for BP3 over the ring-type, respectively. Note that the performance of BP2 over the fully connected pair-wise graph is now worse than that of BP3 over the ring-type. Here, we set the number of iterations of BP2 and BP3 to four and six. One possible reason for why the fully-connected graph perform worse than the ring-type for higher order QAM can be inferred from the convergence behavior as shown in Fig.9, where  it is shown that the convergence for fully-connected pair-wise graph is stuck at three or four iterations and the BER is increased sharply with more number of iterations, while, for the ring-type, it converges steadily.

In Figs.~\ref{Fig:5} to \ref{Fig:8}, the number of iterations was set based on the simulation results in Figs.~\ref{Fig_iter1} and \ref{Fig_iter2}, which we performed with different number of antennas and modulation size, to give insights into how many iterations are required for a satisfactory performance. As shown in the simulation results, the number of iterations required for convergence depends on the modulation sizes, but not much on the number of antennas. Specifically, for BP3, we can say that we need more number of iterations for the eventual convergence with higher modulation size. For BP2, the convergence behavior with different number of antennas looks similar to that of BP3 (specifically for QPSK), while it is quite different from that of BP3 with different modulation sizes. Specifically speaking, the performance of BP2 over the fully connected graph does not get better with more than 3 or 4 iterations. Rather, it is degraded especially for higher order modulation. In BP3 over the ring-type graph, however, no degradation has been observed with more iterations. As mentioned previously, the condition for sure convergence in loopy graph is still an open problem. And the difference in the convergence behavior of BP2 and BP3 can only be explained by the note in \cite{R21}, i.e., in densely connected graph, the messages may circulate along the short loops preventing the eventual convergence. Fig.~\ref{Fig:4}(a) of the fully connected pair-wise graph is more densely connected than Fig.~\ref{Fig:4}(b) of the ring-type pair-wise graph. Although the message will propagate faster in densely connected graph than in sparsely connected graph, resulting in faster convergence, the message circulation may prevents the eventual convergence with more iterations.

Another point we need to note is that, in BP3, one can allow a slight performance degradation for a large computational saving. Certainly, as shown in Fig.~\ref{Fig_iter2} and implicated in Fig.~\ref{Fig:8}, at least 10 iterations is needed for eventual convergence for 16QAM. However, comparing the required SNR for, say, $10^{-4}$ BER, the difference between 6 and 12 iterations is less than 0.1 dB while, in computational burden, 12 iterations is twice that of 6.

Fig.~\ref{Fig_iter3} shows the convergence behavior of the Gaussian BP discussed in Section IV. We plotted the bit error rate performance of the Gaussian BP over the fully-connected and ring-type pair-wise graph, respectively, with various numbers of iterations. As can be seen in the figure, both GBP 2G and 3G converge to the performance of linear MMSE detector, though it requires many more iterations than those of BP2 and BP3. The only difference between GBP 2G and 3G is the rate of convergence. On the other hand, in the high SNR region, the performance appears to worsen with higher SNR. However, it should be noted that with higher SNR eventual convergence simply requires more iterations.%, as also shown in Fig.~\ref{Fig:9}

\section{Conclusions}
In this paper, low complexity, iterative MIMO detection algorithms were derived as a message passing over the pair-wise bipartite graphs with the translation functions that are obtained by marginalizing the posterior joint probability density under the Gaussian input assumption. We investigated two models, the fully-connected and ring-type pair-wise graph. The latter is shown to be an extension of the previous work in~\cite{R9, R10}. The two pair-wise graphical models are rather sparse in the sense that the number of edges connected to an observation node, i.e., edge degree, is only two and, thus, the message passing becomes much easier than that over the fully connected bipartite graph. %The simulation results confirm that the proposed algorithms perform very close to the ML detection. 

We also investigated the proposed algorithm under Gaussian input assumption. It was shown that, for the Gaussian BP over the ring-type pair-wise graph, the mean converges to the linear MMSE estimates, even though the variance converges to a different value from the MMSE obtained by MAP estimation. These results are in line with those in~\cite{R16,R17,R24, R29}. Gaussian BP over the fully-connected pair-wise graph shows a faster convergence rate than Gaussian BP over the ring-type graph. 

As proved in this paper, the convergence of the Gaussian BP 3G over the ring-type graph is guaranteed. This does not, however, appear to be the case for non-Gaussian message. The performance of BP 2 for non-Gaussian case degrade with more than four iterations. This phenomenon might stem from the short cycles in their graphical model and may be avoided by utilizing ``global iteration" between MIMO detection and channel decoding. That is, by employing an appropriate channel code and interleaver, message circulation along local short cycles can be broken up not only for steady convergence but also for better performance. We leave this for our future work.

%\section*{Acknowledgement}
%This work was partly supported by Basic Science Research Program through the National Research Foundation of Korea (NRF) funded by the Ministry of Education, Science and Technology (2009-0064557, 2010-0015582) and Dankook University Project for funding RICT 2011.%The work of C.-B. Chae was in part supported by the Ministry of Knowledge Economy under the �IT Consilience Creative Program� (NIPA-2010-C1515-1001- 0001) and the Yonsei University Research Fund of 2011.
%
%\begin{table}[!t]
%\caption{An approximate number of operations required for detection in a single $M \times M$ MIMO channel (modulation size of $2^m$).}
%\begin{center}
%\begin{tabular}{|c|c|c|c|c|}
%\hline
%Detector & Linear preprocessing & Post (iterative) detection & M=6, m=2, & M=4, m=4\\
%& & (symbol likelihood generation) & $\nu_1$=4, $\nu_2$ = 6 & $\nu_1$=4, $\nu_2$=6 \\ 
%\hline \hline
%MMSE & $24M^3 + 18M^2 + 2M$ & $6M2^m$ & 5,988 & 2,216 \\
%\hline
%ML & 0 & $2^{mM}\cdot (8M^2 + 9M)$ & 1,454,080 & 11,337,728 \\
%\hline
%BP2 & $16M^4+60M^3+437M^2-486M$ & $[2^{2m}\cdot(2\nu_1 + 21) + 2^m(M-1)]\cdot M(M-1)$ & 61,032 & 102,648\\
%\hline
%BP3 & $56M^3 + 113M^2 + 914M$ & $2^{2m} \cdot (2\nu_2 + 21) \cdot 2M$ & 27,984 & 76,632\\ 
%\hline
%\end{tabular}
%\end{center}
%\label{default}
%Parameters ${\nu_1}$ and ${\nu_2}$ are the number of iterations for the BP2 and BP3, respectively.
%\end{table}%

\begin{appendix}[DETAILED DERIVATIONS OF (44) AND THE GAUSSIAN BP]
To derive (\ref{eq45}) and the Gaussian BP rule, (\ref{eq58})-(\ref{eq67}), we use the properties of the Gaussian PDF in~\cite{R21}, some of which are as follows
\begin{align}
1)~&{\mathcal {CN}}\left(x;\mu,{\sigma}^2\right)={\mathcal {CN}}\left(\mu;x,{\sigma}^2\right)
={\mathcal {CN}}\left(x-\mu;0,{\sigma}^2\right)={\mathcal {CN}}\left(\mu-x;0,{\sigma}^2\right) \nonumber \\
2)~&{\mathcal {CN}}\left(ax+b;\mu,{\sigma}^2\right)={\mathcal {CN}}\left(x;\frac{\mu-b}{a},
\frac{{\sigma}^2}{|a|^2}\right) \nonumber\\
\begin{split}
3)~&{\mathcal {CN}}\left(x;\mu_1,{\sigma}_1^2\right)\cdot{\mathcal {CN}}\left(x;\mu_2,{\sigma}_2^2\right) %\nonumber\\
%~&
={\mathcal{CN}}\left(x;\frac{\sigma_1^{-2}\mu_1+\sigma_2^{-2}\mu_2}{\sigma_1^{-2}+\sigma_2^{-2}},
\frac{1}{\sigma_1^{-2}+\sigma_2^{-2}}\right)\cdot{\mathcal{CN}}(\mu_1;\mu_2,\sigma_1^2+\sigma_1^2) \nonumber
\end{split}
\\4)~&\int{\mathcal{CN}}\left(x;\mu_1,\sigma_1^2\right)\cdot{\mathcal{CN}}\left(x;\mu_2,\sigma_2^2\right)
\cdot dx={\mathcal{CN}}\left(\mu_1;\mu_2,\sigma_1^2+\sigma_1^2\right). \nonumber
\end{align}

Using these, (\ref{eq45}) is obtained by direct computation as follows.
\begin{align}
\tilde{p}(x_j|x_i,y_{j|i}^\prime) =& \frac{{\mathcal{CN}}\left(y_{j|i}^\prime;a_{j|i,j}x_j+a_{j|i,i}x_i,\sigma_{j|i}^2\right)
\cdot{\mathcal{CN}}\left(x_j;0,1\right)}
{{\mathcal{CN}}\left(y_{j|i}^\prime;a_{j|i,i}x_i,\sigma_{j|i}^2+|a_{j|i,j}|^2\right)}\nonumber\\
%=&\frac{{\mathcal{CN}}\left(a_{j|i,j}x_j;y_{j|i}^\prime-a_{j|i,i}x_i,\sigma_{j|i}^2\right)
%\cdot{\mathcal{CN}}\left(x_j;0,1\right)}
%{{\mathcal{CN}}\left(y_{j|i}^\prime;a_{j|i,i}x_i,\sigma_{j|i}^2+|a_{j|i,j}|^2\right)}\nonumber\\
=&\frac{{\mathcal{CN}}\left(x_j;\frac{1}{a_{j|i,j}}\left(y_{j|i}^\prime-a_{j|i,i}x_i\right),\frac{\sigma_{j|i}^2}
{|a_{j|i,j}|^2}\right)\cdot{\mathcal{CN}}\left(x_j;0,1\right)}
{{\mathcal{CN}}\left(y_{j|i}^\prime;a_{j|i,i}x_i,\sigma_{j|i}^2+|a_{j|i,j}|^2\right)}\nonumber%\\
\end{align}
\begin{align}
%=&{\mathcal{CN}}\left(x_j;\frac{\frac{|a_{j|i,j}|^2}{a_{j|i,j}\sigma_{j|i}^2}\left(y_{j|i}^\prime-a_{j|i,i}x_i\right)}
%{1+\left(\frac{\sigma_{j|i}^2}{|a_{j|i,j}|^2}\right)^{-1}},
%\left(\left(\frac{\sigma_{j|i}^2}{|a_{j|i,j}|^2}\right)^{-1}+1\right)^{-1}\right)\nonumber %\\
%\cdot\frac{\mathcal{CN}\left(\frac{1}{a_{j|i,j}}\left(y_{j|i}^\prime-a_{j|i,i}x_i\right);0,
%1+\frac{\sigma_{j|i}^2}{|a_{j|i,j}|^2}\right)}{\mathcal{CN}\left(y_{j|i}^\prime;a_{j|i,i}x_i,\sigma_{j|i}^2+|a_{j|i,j}|^2\right)}\nonumber\\
=&\mathcal{CN}\left(x_j;\frac{\frac{1}{a_{j|i,j}}\left(y_{j|i}^\prime-a_{j|i,i}x_i\right)}
{\frac{\sigma_{j|i}^2}{|a_{j|i,j}|^2}+1},
\frac{\sigma_{j|i}^2}{|a_{j|i,j}|^2}\cdot\left(1+\frac{\sigma_{j|i}^2}{|a_{j|i,j}|^2}\right)^{-1}\right)\nonumber %\\
\cdot\frac{\mathcal{CN}\left(y_{j|i}^\prime-a_{j|i,i}x_i;0,|a_{j|i,j}|^2+\sigma_{j|i}^2\right)}
{\mathcal{CN}\left(y_{j|i}^\prime;a_{j|i,i}x_i,\sigma_{j|i}^2+|a_{j|i,j}|^2\right)} \nonumber \\
=&\mathcal{CN}\left(x_j;\frac{a_{j|i,j}^*}{\sigma_{j|i}^2+|a_{j|i,j}|^2}\left(y_{j|i}^\prime-a_{j|i,i}x_i\right),
\frac{\sigma_{j|i}^2}{\sigma_{j|i}^2+|a_{j|i,j}|^2}\right). \label{eq:Ap5}
\end{align}
\newline
Now, we derive the message update rule of the Gaussian BP. To this end, we divide the message update rule in BP2, (\ref{eq46}), into two steps, i.e., 
the extrinsic information computation, $\lambda_{i \to j}(x_i) = \prod\nolimits_{k \in V\left( i \right)\backslash j} {\pi _{k \to i} } \left( {x_i } \right)$,
and the message translation step, 
$\pi _{i \to j} (x_j ) = \alpha \sum\nolimits_{x_i  \in \Xi } {\tilde{p}(x_j | x_i, \pmb{y} )} \lambda_{i \to j}(x_i) $. 
Assuming the Gaussian messages,
%\begin{eqnarray}
$\pi _{k \to i} (x_i ) = \mathcal{CN}( {x_i ;\mu _{\pi ,k \to i}^{} ,{\rm{ }}\sigma _{\pi ,k \to i}^2 } )$, % \label{eqa06} 
%\end{eqnarray}
the former is given by
%\begin{eqnarray}
\begin{align}
\lambda_{i \to j}(x_i) &= \prod\limits_{k \in V(i)\backslash j} {\pi _{k \to i} (x_i )} 
 = \prod\limits_{k \in V(i)\backslash j} {\mathcal{CN}\left( {x_i ;\mu _{\pi ,k \to i}^{} ,{\rm{ }}\sigma _{\pi ,k \to i}^2 } \right)}  \nonumber \\
& \propto \mathcal{CN}\left( {x_i ;\frac{{\sum\nolimits_{k \in V(i)\backslash j} {\sigma _{\pi ,k \to i}^{ - 2} \mu _{\pi ,k \to i}^{} } }}{{\sum\nolimits_{k \in V(i)\backslash j} {\sigma _{\pi ,k \to i}^{ - 2} } }},}\right. 
    {{{\left( {\sum\limits_{k \in V(i)} {\sigma _{\pi ,k \to i}^{ - 2}} } \right)}^{ - 1}}}\Bigg)  %\nonumber \\
= \mathcal{CN}\left( {x_i ;\mu _{\lambda ,i \to j}^{} ,{\rm{ }}\sigma _{\lambda ,i \to j}^2 } \right) \label{eqa08}
\end{align}
%\end{eqnarray}
%where we set
%\begin{align}
%\mu _{\lambda ,i \to j} = \frac{ {\sum\nolimits_{k \in V(i)\backslash j} {\sigma _{\pi ,k \to i}^{ - 2} \mu _{\pi ,k \to i} } }} %{{\sum\nolimits_{k \in V(i)\backslash j} {\sigma _{\pi ,k \to i}^{ - 2} } }},  \;
% \sigma _{\lambda ,i \to j}^2 =  {{{\left( {\sum\limits_{k \in V(i)} {\sigma _{\pi ,k \to i}^{ - 2}} } \right)}^{ - 1}}}.  \nonumber
%\end{align}
For the message translation, we first rewrite (\ref{eq:Ap5}) as
%\begin{eqnarray}
\begin{align}
\tilde{p}(x_j \left| {x_i ,y'_{j|i} } \right.)  \nonumber &
= \mathcal{CN}\left( {x_j ;\frac{{a_{j|i,j}^* }}{{\sigma _{j|i}^2  + |a_{j|i,j} |^2 }}\left( {y'_{j|i}  - a_{j|i,i} x_i } \right),}\right. \frac{{\sigma _{j|i}^2 }}{{\sigma _{j|i}^2  + |a_{j|i,j} |^2 }}\Bigg)  \nonumber \\
&=\mathcal{CN}\left( {\frac{{\sigma _{j|i}^2  + |a_{j|i,j} |^2 }}{{a_{j|i,j}^* }}x_j ;\left( {y'_{j|i}  - a_{j|i,i} x_i } \right),}\right.
\frac{{\sigma _{j|i}^2 (\sigma _{j|i}^2  + |a_{j|i,j} |^2 )}}{{|a_{j|i,j} |^2 }} \Bigg)  \nonumber \\
&= \mathcal{CN}\left( {x_i ;\frac{1}{{a_{j|i,i} }}\left( {y'_{j|i}  - \frac{{\sigma _{j|i}^2  + |a_{j|i,j} |^2 }}{{a_{j|i,j}^* }}x_j } \right),}\right.
\frac{{\sigma _{j|i}^2 (\sigma _{j|i}^2  + |a_{j|i,j} |^2 )}}{{|a_{j|i,j} |^2 |a_{j|i,i} |^2 }} \Bigg). \label{eqa10}
\end{align}
%\end{eqnarray}
Then, by plugging (\ref{eqa08}) and (\ref{eqa10}) into (\ref{eq46}) and changing the summation into integral,\footnote{Since the
input is now continuous Gaussian random variable, we need to change the summation in (\ref{eq45}) into integral.} we have
\begin{align}
%\begin{split}
\pi&_{i\rightarrow j}(x_j)=\int_{x_i}p\left(x_j|x_i,y_{ji}^\prime\right)\cdot\lambda_{i\rightarrow j}(x_i)\cdot dx_i \nonumber \\
=&\int_{x_i}\mathcal{CN}\left(x_i;\frac{1}{a_{j|i,i}}\left(y_{j|i}^\prime-\frac{\sigma_{j|i}^2+|a_{j|i,j}|^2}{a_{j|i,j}^*}x_j\right),
\frac{\sigma_{j|i}^2\left(\sigma_{j|i}^2+|a_{j|i,j}|^2\right)}{|a_{j|i,j}|^2|a_{j|i,i}|^2}\right)
\cdot\mathcal{CN}\left(x_i;\mu_{\lambda,i\rightarrow j},\sigma_{\lambda,i\rightarrow j}^2\right)\cdot dx_i \nonumber
\\%\end{align}
%
%\begin{align}
%=&\mathcal{CN}\left(\frac{1}{a_{j|i,i}}\left(y_{j|i}^\prime-\frac{\sigma_{j|i}^2+|a_{j|i,j}|^2}{a_{j|i,j}^*}x_j\right);
%\mu_{\lambda,i\rightarrow j},
%\frac{\sigma_{j|i}^2\left(\sigma_{j|i}^2+|a_{j|i,j}|^2\right)}{|a_{j|i,j}|^2|a_{j|i,i}|^2}+\sigma_{\lambda,i\rightarrow %j}^2\right) \nonumber\\
&=\mathcal{CN}\left(y_{j|i}^\prime-\frac{\sigma_{j|i}^2+|a_{j|i,j}|^2}{a_{j|i,j}^*}x_j-a_{j|i,i}\mu_{\lambda,i\rightarrow j};0,
\frac{\sigma_{j|i}^2\left(\sigma_{j|i}^2+|a_{j|i,j}|^2\right)}{|a_{j|i,j}|^2}+|a_{j|i,i}|^2\sigma_{\lambda,i\rightarrow j}^2\right) \nonumber\\
&=\mathcal{CN}\left(x_j;\frac{a_{j|i,j}^*}{\sigma_{j|i}^2+|a_{j|i,j}|^2}\left(y_{j|i}^\prime-a_{j|i,i}\mu_{\lambda,i\rightarrow j}
\right),\frac{\sigma_{j|i}^2}{\sigma_{j|i}^2+|a_{j|i,j}|^2}+\frac{|a_{j|i,j}|^2 |a_{j|i,i}|^2\sigma_{\lambda,i\rightarrow j}^2}{\left(\sigma_{j|i}^2+|a_{j|i,j}|^2\right)^2}
\right) \label{eqa11}\\
&=\mathcal{CN}\left(x_j;\mu_{\pi,i\rightarrow j},\sigma_{\pi,i\rightarrow j}^2\right). \label{eqa12}
%\end{split}
\end{align}
By comparing the mean and variance in (\ref{eqa11}) and (\ref{eqa12}), we obtain the message passing rules of (\ref{eq58}) 
and (\ref{eq59}), respectively. The belief in (\ref{eq60}) and  (\ref{eq61}) can be obtained similarly to the derivation in (\ref{eqa08}).

\end{appendix}

\renewcommand{\baselinestretch}{1.0}
\bibliographystyle{IEEEbib}

\bibliography{references_combined}
\begin{table}[!t]
\caption{An approximate number of operations required for detection in a single $M \times M$ MIMO channel (modulation size of $2^m$).}
\begin{center}
\begin{tabular}{|c|c|c|c|c|}
\hline
Detector & Linear preprocessing & Post detection & M=6, m=2, & M=4, m=4\\
& &  & $\nu_1$=4, $\nu_2$ = 6 & $\nu_1$=4, $\nu_2$=6 \\ 
\hline \hline
MMSE & $24M^3 + 18M^2 + 2M$ & $6M\cdot2^m$ & 5,988 & 2,216 \\
\hline
ML & 0 & $2^{mM}\cdot (8M^2 + 9M)$ & 1,454,080 & 11,337,728 \\
\hline
BP2 & $16M^4+60M^3+437M^2-486M$ & $[2^{2m}\cdot(2\nu_1 + 21) + 2^m\nu_1]\cdot M(M-1)$ & 61,032 & 102,648\\
\hline
BP3 & $56M^3 + 113M^2 + 914M$ & $2^{2m} \cdot (2\nu_2 + 21) \cdot 2M$ & 27,984 & 76,632\\ 
\hline
\end{tabular}
\end{center}
\label{default}
Parameters ${\nu_1}$ and ${\nu_2}$ are the number of iterations for the BP~2 and BP~3, respectively.
\end{table}%

%\begin{table}[!b]
%\caption{An approximate number of operations required for detection in a single $M \times M$ MIMO channel (modulation size of $2^m$).}
%\begin{center}
%\begin{tabular}{|c|c|c|c|c|}
%\hline
%Detector & Linear preprocessing & Post (iterative) detection & M=6, m=2, & M=4, m=4\\
%& & (symbol likelihood generation) & $\nu_1$=4, $\nu_2$ = 6 & $\nu_1$=4, $\nu_2$=6 \\ 
%\hline \hline
%MMSE & $24M^3 + 18M^2 + 2M$ & $6M2^m$ & 5,988 & 2,216 \\
%\hline
%ML & 0 & $2^{mM}\cdot (8M^2 + 9M)$ & 1,454,080 & 11,337,728 \\
%\hline
%BP2 & $16M^4+60M^3+437M^2-486M$ & $[2^{2m}\cdot(2\nu_1 + 21) + 2^m(M-1)]\cdot M(M-1)$ & 61,032 & 102,648\\
%\hline
%BP3 & $56M^3 + 113M^2 + 914M$ & $2^{2m} \cdot (2\nu_2 + 21) \cdot 2M$ & 27,984 & 76,632\\ 
%\hline
%\end{tabular}
%\end{center}
%\label{default}
%Parameters ${\nu_1}$ and ${\nu_2}$ are the number of iterations for the BP2 and BP3, respectively.
%\end{table}%

%\newpage
\begin{figure}[!t]
  \centerline{\resizebox{0.65\columnwidth}{!}{\includegraphics{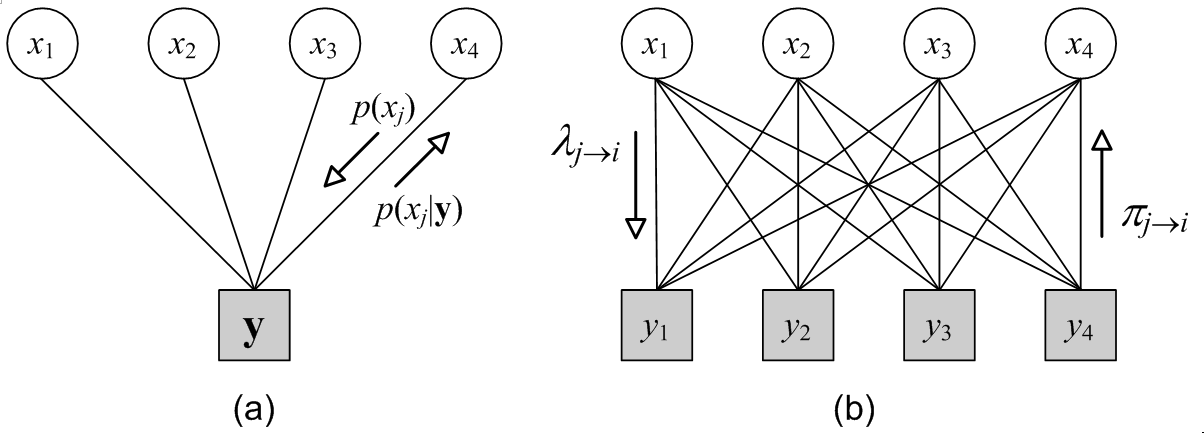}}}
   \caption{Bipartite graphs for a $4 \times 4$ MIMO channel. The circles are variable nodes corresponding to a data symbol and the boxes labeled by ${\pmb y}$ and $y_j$ are observation nodes corresponding to the received signal.}
   \label{Fig:1}
\end{figure}
%\begin{figure}[!t]
% \centerline{\resizebox{0.55\columnwidth}{!}{\includegraphics{MRF_Fig3.png}}}
%   \caption{Fully-connected (a), ring-type (b) Markov random fields with 4 variables.}
%   \label{Fig:3}
%\end{figure}
\begin{figure}[!t]
  \centerline{\resizebox{0.55\columnwidth}{!}{\includegraphics{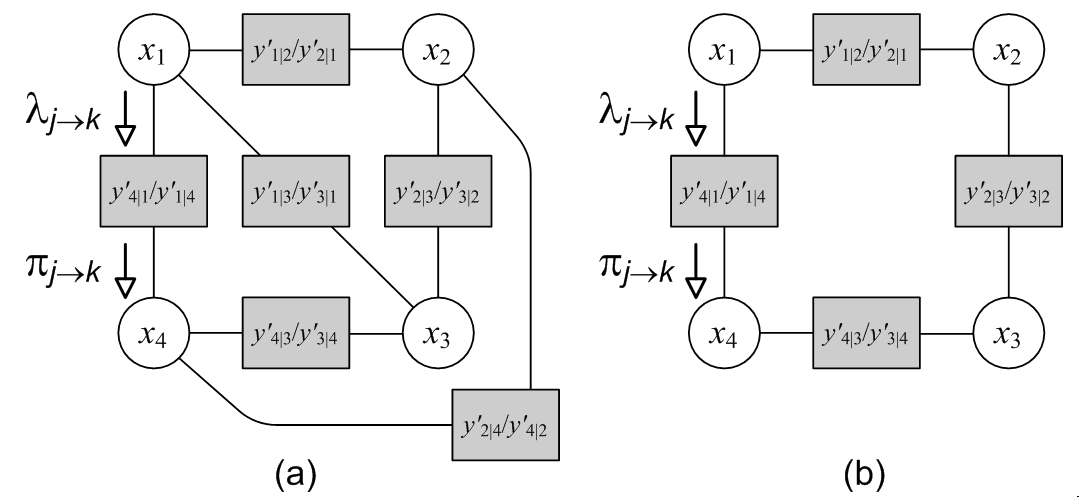}}}
   \caption{The bipartite graph for (a) fully-connected pair-wise model and (b) ring-type pair-wise model, respectively, for a $4 \times N$ MIMO channel.}
   \label{Fig:4}
\end{figure}

\begin{figure}[!t]
  \centerline{\resizebox{0.7\columnwidth}{!}{\includegraphics{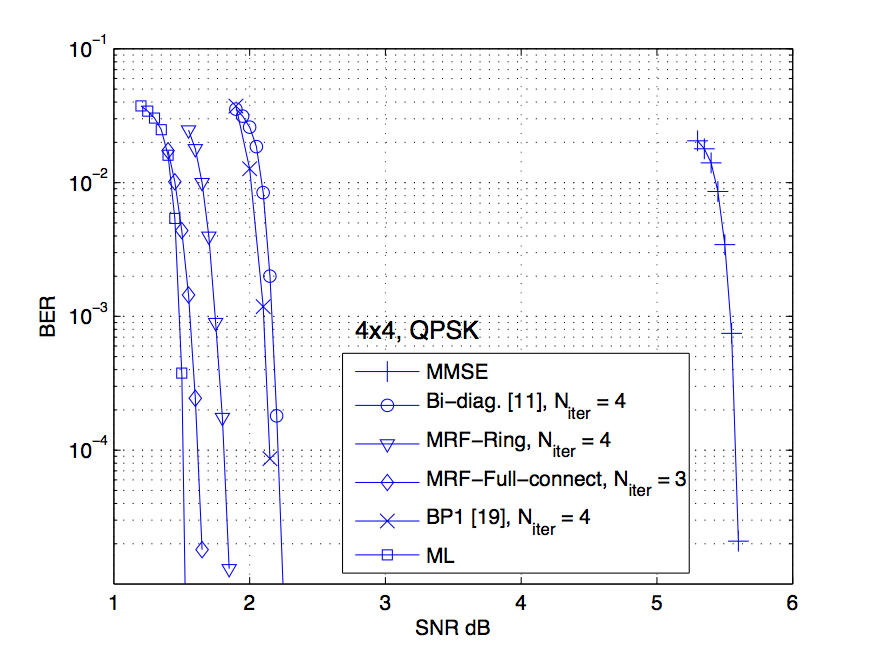}}}
   \caption{A comparison of bit error rate performance of MMSE, MAP and the proposed detectors as a function of SNR $(1/\sigma^2)$; $4\times 4$ antenna configuration, QPSK modulation with DVB-S2 LDPC code of rate 3/4 (length 64800).}
   \label{Fig:5}
\end{figure}

\begin{figure}[!t]
  \centerline{\resizebox{0.7\columnwidth}{!}{\includegraphics{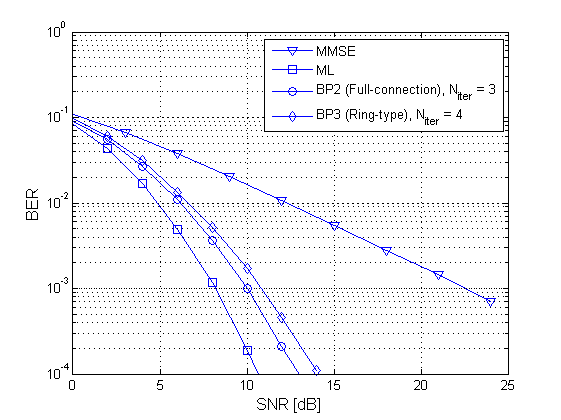}}}
   \caption{A comparison of bit error rate performance of MMSE, MAP, and the proposed detectors as a function of SNR $(1/\sigma^2)$; $4\times 4$ antenna configuration, QPSK modulation, no channel coding.}
   \label{Fig:8a}
\end{figure}

\begin{figure}[!t]
  \centerline{\resizebox{0.7\columnwidth}{!}{\includegraphics{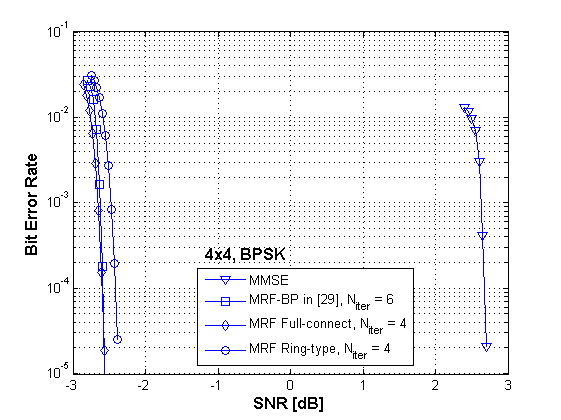}}}
   \caption{A comparison of bit error rate performances of MMSE, the proposed detectors and the detector in \cite{R33} with a damping factor 0.45.}
   \label{Fig:r2}
\end{figure}

\begin{figure}[!t]
  \centerline{\resizebox{0.7\columnwidth}{!}{\includegraphics{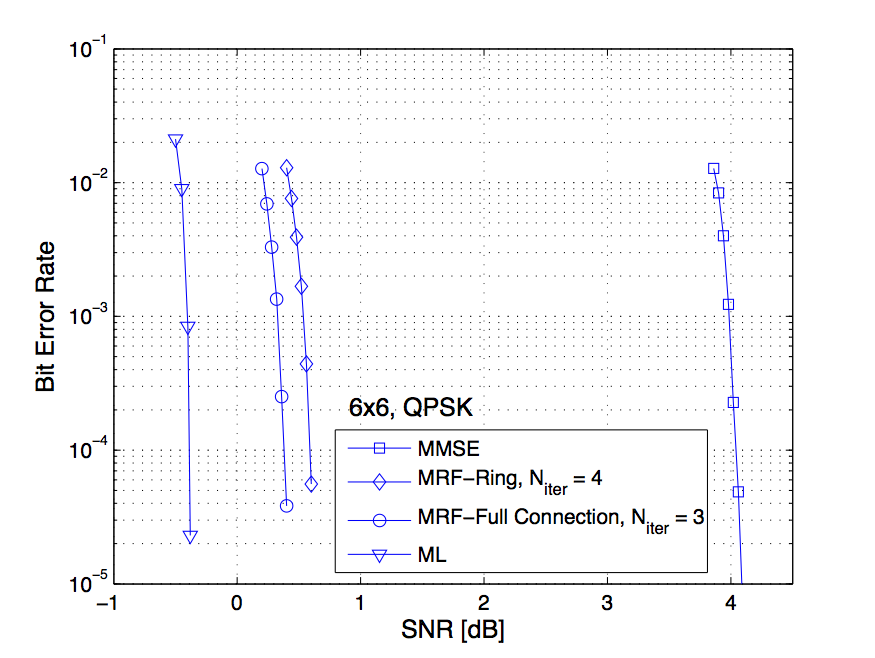}}}
   \caption{A comparison of bit error rate performance of MMSE, MAP and the proposed detectors as a function of SNR $(1/\sigma^2)$; $6\times 6$ antenna configuration, QPSK modulation with DVB-S2 LDPC code of rate 3/4 (length 64800).}
   \label{Fig:6}
\end{figure}

\begin{figure}[!t]
  \centerline{\resizebox{0.7\columnwidth}{!}{\includegraphics{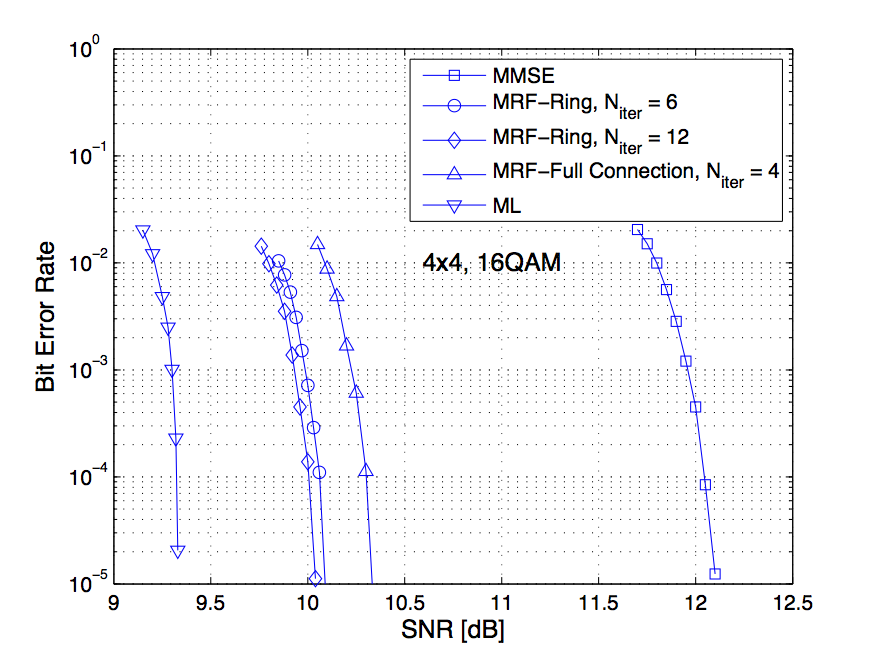}}}
   \caption{A comparison of bit error rate performance of MMSE, MAP, and the proposed detectors as a function of SNR $(1/\sigma^2)$; $4\times 4$ antenna configuration, 16QAM modulation with DVB-S2 LDPC code of rate 3/4 (length 64800).}
   \label{Fig:8}
\end{figure}

%\begin{figure}[!t]
%  \centerline{\resizebox{0.7\columnwidth}{!}{\includegraphics{Fig8.eps}}}
%   \caption{Bit error rate performance of BP 1 over fully connected MRF as a function of SNR $(1/\sigma^2)$; $4\times 4$ antenna configuration, QPSK modulation with DVB-S2 LDPC code of rate 3/4 (length 64800), the number of iterations = 2, 3, 4 and 6, respectively.}
%   \label{Fig:7}
%\end{figure}

\begin{figure}[!t]
  \centerline{\resizebox{1\columnwidth}{!}{\includegraphics{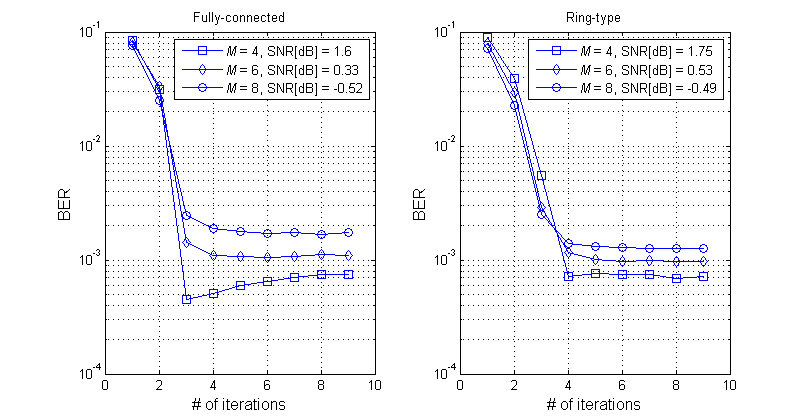}}}
   \caption{Convergence property of the proposed algorithm with the number of antennas $M =$ 4, 6 and 8; QPSK modulation, DVB-S2 LDPC code of rate 3/4 (length 64800).}
   \label{Fig_iter1}
\end{figure}

\begin{figure}[!t]
  \centerline{\resizebox{1\columnwidth}{!}{\includegraphics{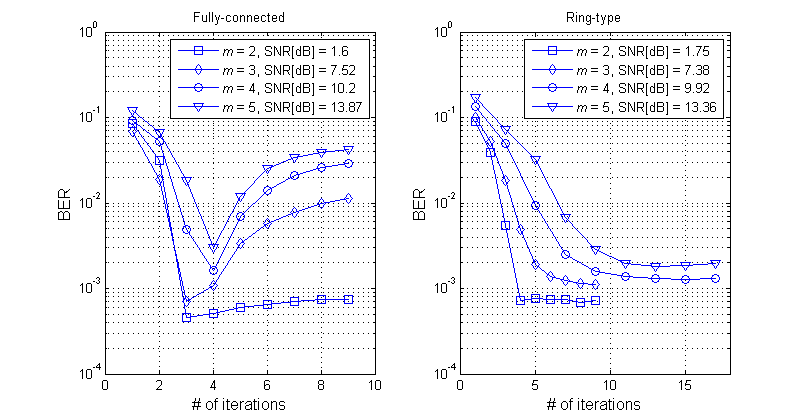}}}
   \caption{Convergence property of the proposed algorithm with modulation size of $2^m$ for $m=$ 2, 3, 4 and 5;  $4\times 4$ antenna configuration, DVB-S2 LDPC code of rate 3/4 (length 64800)}
   \label{Fig_iter2}
\end{figure}

%\begin{figure}[!t]
%  \centerline{\resizebox{0.67\columnwidth}{!}{\includegraphics{Re_Fig8.png}}}
%   \caption{Convergence characteristics of the Gaussian BP over the fully-connected and ring-type MRF, respectively; $4\times 4$ antenna configuration, SNR = 5, 20 dB. }
%   \label{Fig:9}
%\end{figure}

\begin{figure}[!t]
  \centerline{\resizebox{1\columnwidth}{!}{\includegraphics{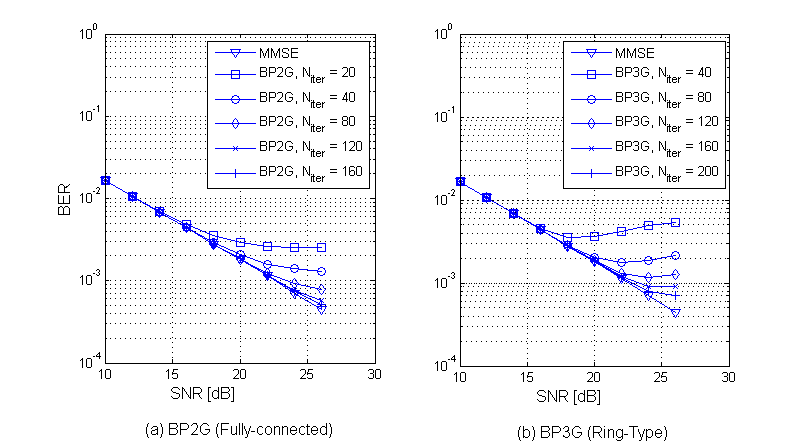}}}
   \caption{Bit error rate performance of the Gaussian BP over the fully-connected and ring-type pair-wise graph, respectively; $4\times 4$ antenna configuration, QPSK modulation, no channel coding.}
   \label{Fig_iter3}
\end{figure}

\end{document}